\documentclass[sn-mathphys,Numbered]{sn-jnl}

\usepackage{graphicx}%
\usepackage{multirow}%
\usepackage{amsmath,amssymb,amsfonts}%
\usepackage{amsthm}%
\usepackage{mathrsfs}%
\usepackage[title]{appendix}%
\usepackage{xcolor}%
\usepackage{textcomp}%
\usepackage{manyfoot}%
\usepackage{booktabs}%
\usepackage{algorithm}%
\usepackage{algorithmicx}%
\usepackage{algpseudocode}%
\usepackage{listings}%
\usepackage{booktabs}%
\usepackage{algcompatible}
\usepackage{listings}%
\usepackage{subcaption}
\usepackage{times}
\usepackage{soul}
\usepackage{url}
\usepackage{eucal}
\usepackage{hyperref}
\usepackage{longtable}
\usepackage{array,multirow}
\usepackage{microtype}
\usepackage[inline]{enumitem}
\usepackage{hhline}

\def \MBR{\mathbb{M}_{r,BFS}}

\def \MB{\mathbb{M}_{BFS}}
\def \MD{\mathbb{M}_{DFS}}
\def \MG{\mathbb{M}_{AP}}

\def \AA{\mathcal{A}}

\def \qq{\boldsymbol{q}}
\def \bb{\boldsymbol{b}}

\def \safepolicypp{\emph{First-Come-First-Served policy}}
\def \safepolicyp{FCFS policy}
\def \safepoliciesp{FCFS policies}
\def \safepolicy{\safepolicyp }
\def \safepolicies{\safepoliciesp }

\def \SP{FCFSP}

\newtheorem{lemma}{Lemma}

\newtheorem{theorem}{Theorem}
\newtheorem{example}{Example}%

\theoremstyle{thmstylethree}%

\begin{document}

\title[Edge Manipulations for the Maximum Vertex-Weighted Bipartite $b$-matching]{Edge Manipulations for the Maximum Vertex-Weighted Bipartite $b$-matching}


\author*[1]{\fnm{Gennaro} \sur{Auricchio}}\email{ga647@bath.ac.uk}

\author[2]{\fnm{Qun} \sur{Ma}}

\author[1]{\fnm{Jie} \sur{Zhang}}

\affil*[1]{\orgdiv{Department of Computer Science}, \orgname{University of Bath}, \orgaddress{\street{Claverton Down}, \city{Bath}, \postcode{BA2 7AY}, \state{England}, \country{United Kingdom}}}

\affil[2]{\orgdiv{College of Intelligence and Computing, Peiyang Campus}, \orgname{Tianjin University}, \orgaddress{\street{135 Yaguan Road, Haihe Education Park}, \city{Tianjin}, \postcode{300354}, \country{China}}}



\abstract{In this paper, we explore the Mechanism Design aspects of the Maximum Vertex-weighted $b$-Matching (MVbM) problem on bipartite graphs $(A\cup T, E)$.
The set $A$ comprises agents, while $T$ represents tasks.
The set $E$, which connects $A$ and $T$, is the private information of either agents or tasks.
In this framework, we investigate three mechanisms -- $\MB$, $\MD$, and $\MG$ -- that, given an MVbM problem as input, return a $b$-matching.
We examine scenarios in which either agents or tasks are strategic and report their adjacent edges to one of the three mechanisms.
In both cases, we assume that the strategic entities are bounded by their statements: they can hide edges, but they cannot report edges that do not exist.
First, we consider the case in which agents can manipulate. 
In this framework, $\MB$ and $\MD$ are optimal but not truthful.
By characterizing the Nash Equilibria induced by $\MB$ and $\MD$, we reveal that both mechanisms have a Price of Anarchy ($PoA$) and Price of Stability ($PoS$) of $2$.
These efficiency guarantees are tight; no deterministic mechanism can achieve a lower $PoA$ or $PoS$.
In contrast, the third mechanism, $\MG$, is not optimal, but truthful and its approximation ratio is $2$.
We demonstrate that this ratio is optimal; no deterministic and truthful mechanism can outperform it.
We then shift our focus to scenarios where tasks can exhibit strategic behaviour.
In this case, $\MB$, $\MD$, and $\MG$ all maintain truthfulness, making $\MB$ and $\MD$ truthful and optimal mechanisms. 
In conclusion, we investigate the manipulability of $\MB$ and $\MD$ through experiments on randomly generated graphs.
We observe that \begin{enumerate*}[label=(\roman*)]
    \item $\MB$ is less prone to be manipulated by the first agent than $\MD$;
    \item $\MB$ is more manipulable on instances in which the total capacity of the agents is equal to the number of tasks;
    \item randomizing the agents' order reduces the agents' ability to manipulate $\MB$.
\end{enumerate*}}

\keywords{Mechanism Design, Game Theory, Matching Problem, Price of Anarchy, Price of Stability}



\maketitle

\section{Introduction}

Managers of large companies are periodically required to submit a list of the projects they carried out for external evaluations.
Alongside the list of projects, a manager needs to identify a worker liable for the achievement of every project.
Due to the workload cap, every worker can be found responsible for only a finite number of projects.
Moreover, every project has a different prestige, which can be regarded as its overall value. 
The manager aims to maximize the total value of the reported projects by deploying its staff following the aforementioned constraints.
Meanwhile, every worker is interested in being associated with high-value projects rather than low-value ones. 
In this context, workers might benefit from hiding their connections to low-value projects.
Indeed, by concealing certain connections, the worker might be able to compel the manager to match them to higher-valued projects in order to increase the overall quality of the report.
A similar situation occurs in the universities of several European countries.
Indeed, to assess the research impact of their higher education institutions, the country asks the university for a list of their best publications, along with the name of the main author.\footnote{This, for example, is a common practice in the United Kingdom, see \url{https://www.ref.ac.uk} for a reference.}
On one hand, every university wants to find a matching between its lecturers and the publications that maximizes the total impact score.
On the other hand, individual academics want to be designated as the main author of their best publications.
In both scenarios, we need to allocate a set of resources that possess an objective and publicly known value, be they projects and their prestige or publications and their impact score, to a set of self-interested entities.
Aside from these two examples, there are many other real-life situations that can be rephrased as matching problems with self-interested agents.
Matching problems were first introduced to minimize transportation costs \cite{Hitchcock1941,Kantorovitch1958} and to optimally assign workers to job positions \cite{Easterfield1946,Thorndike1950}. 
Thereafter, bipartite matching found application in several applied problems, such as sponsored searches \cite{DBLP:journals/sigact/BirnbaumM08,DBLP:conf/www/AggarwalMPP09,DBLP:journals/fttcs/Mehta13}, school admissions \cite{abdulkadirouglu2005college,DBLP:journals/tcs/BiroFIM10}, scheduling \cite{ZHAO2010837,DBLP:conf/soda/GuptaKS14}, peer-reviewing \cite{charlin2013toronto}, and general resource allocation \cite{feng2013device,7295474}.
Indeed, characterizing matching through the maximization or minimization of a functional allows us to define mathematical objects that arise in several applied contexts. 
For example, Maximum Cardinality Matchings have connections with the computation of perfect matchings \cite{fukuda1994finding}, bottleneck matchings \cite{edmonds1970bottleneck}, and Lévy-Prokhorov distances, \cite{lahn2021faster}.
Another example is the Maximum Edge-weighted Matching problem, which has been widely used in clustering problems \cite{ding2019clustering}, machine learning \cite{bayati2005maximum}, and also to compute Wasserstein-barycenters \cite{auricchio2019computing}.
For a complete discussion of the matching problems and their applications, we refer to \cite{lovasz2009matching}.

{\bf Game-theoretic aspects of matching problems.}
In many matching problems, the bipartite graph and the weights of its edges describe the preferences that a group of agents has over the members of a second, disjointed, group. 
This happens both in one-sided matching (as, for example, house allocation \cite{HZ79,10.1257/000282802762024728}  or the resident/hospital problem \cite{manlove2013algorithmics}) and in two-sided matching (as, for example, stable marriage \cite{GaleS13}, student admission \cite{balinski1999tale}, and market matching, in general, \cite{kojima2019new}). 
Starting from the basic models, several variants have been proposed, most of the time, these variants enforce some meaningful constraint, such as regional constraints \cite{kamada2017recent}, diversity constraints \cite{ehlers2014school}, or lower quotas \cite{yokoi2017generalized}.
In all these frameworks, however, the preferences of each agent are reported by the agent itself. Therefore, studying the social aspects of implementing a mechanism to find an allocation is essential.
To this extent, several works studied the manipulability of these class of problems \cite{fragiadakis2016strategyproof,krysta2014size}, the fairness \cite{kamada2019fair}, and the envy-freeness \cite{abdulkadirouglu2003school}.
We refer to \cite{aziz2022matching} for a comprehensive survey.
In this paper, we deal with a different game-theoretical problem in which every strategic entity has the same preference order over the other side of the bipartite graph.
Given that the preference orders over the other side of the bipartite graph are common and public information, we assume that the private information of the strategic entities consists of the set of edges connecting them to the other side of the graph, which is the main novelty of our framework.
In this framework, we study both truthful and optimal mechanisms and provide positive and tight results on the efficiency guarantees.
A previous version of this paper appeared in \cite{gaecai23}.

\subsection*{Our Contributions}

Throughout the paper, we assume that the two sides of the bipartite graph consisting of a set of agents $A=\{a_1,\dots,a_n\}$ and a set of tasks $T=\{t_1,\dots,t_m\}$.
We represent $E$ as the set of edges connecting agents to tasks.
Each task $t_j$ is assigned a positive value, denoted as $q_j$, which is equal for all agents.
Moreover, each task can be matched with at most one agent in set $A$.
Each agent $a_i$ is associated with a positive integer, denoted as $b_i$, which specifies the maximum number of tasks that agent $a_i$ can be allocated.
We consider two game-theoretical frameworks.
In the first framework, the agents' set is composed of strategic entities (see Section \ref{sec:agentmanipulation}), while in the second framework, the tasks are behaving strategically (see Section \ref{sec:tasksmanipulation}).
Notice that, since the roles of agents and tasks in the MVbM problem are not interchangeable, the two frameworks are inherently different and yield different results.
Our main focus, however, pertains to the study of the first framework, i.e. the case in which the agents are reporting their information, which is also more challenging.
In both frameworks, we study three mechanisms induced by known algorithms.
Two of them, namely $\MB$ and $\MD$, are induced by the algorithm proposed in \cite{spencer1984node}.
The third one, namely $\MG$, is induced by the algorithm proposed in \cite{dobrian20192}, which is an approximation of the one proposed in \cite{spencer1984node}.
Throughout the paper, we assume that the agents (or tasks) are bounded by their statements, hence they can hide edges or part of their capacity (value), but they cannot report non-existing edges or a capacity (value) higher than the real one.
{\bf Agents' Manipulations -- The properties of $\MB$ and $\MD$.}
First, we study the case in which the agents are able to hide part of the edges connecting them to tasks and show that any mechanism that returns an MVbM cannot be truthful.
In particular, we infer that both $\MB$ and $\MD$ are not truthful.
We then show that the first agent has a strategic advantage over the other agents and that, in general, the individual agent's ability to alter the outcome of the mechanisms depends on the order in which the mechanism processes the agents' set.
Due to this correlation, we are able to describe the Nash Equilibria induced by the mechanisms and prove that the Price of Anarchy (PoA) and Price of Stability (PoS) of both $\MB$ and $\MD$ are equal to $2$.
We then proceed to show that $2$ is the best possible PoA and PoS achievable by deterministic mechanisms, thus $\MB$ and $\MD$ are the best possible mechanisms in terms of PoA and PoS.
Afterwards, we characterize some sets of inputs on which these mechanisms become truthful. 
Lastly, we extend our study to the case in which the agents report both the edges connecting them to the set of tasks and their capacity.
{\bf Agents' Manipulations --  The properties of $\MG$.}
Although this mechanism does not return an MVbM and hence is not optimal, we show that $\MG$ is truthful for agents' manipulations regardless of whether the agents can report only the edges or the edges and the capacity. 
Furthermore, it is group strategyproof whenever all the tasks have different values.
We show that its approximation ratio is $2$ and show that this bound is tight.
Thus, there does not exist a deterministic truthful mechanism that has a lower approximation ratio.
Moreover, we prove that the output produced by the approximation mechanism is actually the matching returned by $\MB$ and $\MD$ in their worst Nash Equilibrium.
{\bf Tasks' Manipulations.}
We then extend our study on the properties of the mechanisms when the tasks are allowed to act strategically.
In this case, the set of incident edges and/or the values are the tasks' private information.
In this simpler framework, we have that all of the three mechanisms we considered are truthful, albeit not group strategyproof.
We summarize the results we find of these mechanisms for both the Tasks and Agents Manipulation cases in Table \ref{table:schedule}.

{\bf Empirical Studies.} In Section \ref{sec:exp_res}, we empirically compare
the manipulability of $\MB$ and $\MD$ over randomly generated instances.
In particular, we observe that:
 \begin{enumerate*}[label=(\roman*)]
    \item $\MB$ is less prone to be manipulated by the first agent than $\MD$. Moreover, the number of instances in which the first agent is able to manipulate $\MB$ decreases as the number of agents increases, while the same does not hold for $\MD$.
    \item $\MB$ is more manipulable on instances in which the total capacity of the agents is equal to the number of tasks.
    \item Randomizing the processing order of $\MB$ reduces the overall manipulability of the mechanism. Thus knowing the processing order of the agents is essential information to manipulate $\MB$.
\end{enumerate*}


\begin{table*}[t!]
\centering
\begin{tabular}{c | c | c | c | c }
\toprule[2pt]
          \multicolumn{5}{c}{Agent Manipulation (Edge Manipulation$\slash$Edge and Capacity Manipulation)     }   \\
 \midrule
    & Truthful & Group SP  & Optimal  & Efficiency \\
 \hhline{~|-|-|-|-}
 $\MB$   & No & No & Yes & PoA = PoS = 2  \\
   $\MD$  & No & No & Yes & PoA = PoS = 2 \\
 $\MG$   & Yes & Yes* & No & $ar$ = 2 \\


\midrule[1.5pt]

\multicolumn{5}{c}{Task Manipulation (Edge Manipulation$\slash$Edge and Value Manipulation)} \\
\midrule
    & Truthful & Group SP  & Optimal  & Efficiency   \\
\hhline{~|-|-|-|-}
 $\MB$   & Yes & No & Yes & $ar$ = 1\\
   $\MD$  & Yes & No & Yes & $ar$ = 1 \\
 $\MG$   &  Yes & No & No & $ar$ = 2 \\
\bottomrule[2pt]
\end{tabular}
\caption{The mechanisms we study and their properties in the two game-theoretical settings we consider. The "Yes*" indicates that the mechanism is group strategyproof under minor assumptions. \label{table:schedule}}
\end{table*}  

\section{Preliminaries}
\label{sect:prelim}

In this section, we recall the basic notions and notations on the MVbM problem. 
Then, we describe the algorithm that defines the mechanisms we study.

{\bf The Maximum Vertex-weighted $b$-Matching Problem.}
Let $G=(A \cup T,E)$ be a bipartite graph. 
Throughout the paper we refer to $A=\{a_1,a_2, \cdots, a_{n}\}$ as the set of agents and refer to $T=\{t_1,t_2, \cdots, t_{m}\}$ as the set of tasks.
The set $E$ contains the edges of the bipartite graph.
We say that an edge $e\in E$ belongs to agent $a_i$ if $e=(a_i,t_j)$ for a $t_j\in T$.
We denote with $M=|E|$ and $N=|A \cup T|=n+m$ the number of edges and the total number of vertices of the graph, respectively. 
Since the graphs are undirected, we denote an edge by $(a_i,t_j)$ and $(t_j,a_i)$ interchangeably. 
Moreover, for any given agent $a_i\in A$, we define the set $T_{E,i}$ as the set of tasks in $T$ that are connected to agent $a_i$ through the set of edges $E$.
When it is clear from the context which set $E$ we are referring to, we drop the subscript from $T_{E,i}$ and use $T_i$.
Let $\bb=(b_1,b_2,\dots,b_n)$ be the vector containing the capacities of the agents, where $b_i\in \mathbb{N}$ for every $i=1,\dots,n$.
A subset $\mu\subset E$ is a $b$-matching if, for every vertex $a_{i}\in A$, the number of edges in $\mu$ linked to $a_i$ is less than or equal to $b_i$ and, for every vertex $t_j\in T$, the number of edges in $\mu$ linked to $t_j$ is, at most, one.
Given a $b$-matching $\mu$, the vertex $a_i \in A$ is \emph{saturated} with respect to $\mu$ if the number of edges in $\mu$ linked to $a_i$ is exactly $b_i$.
Otherwise, the vertex $a_i$ is \emph{unsaturated} with respect to $\mu$. 
We denote with $\qq=(q_1,q_2,\dots,q_m)$ the vector containing the values of the tasks, in our framework, we have $q_j>0$ for every $j$. 
The value of a matching $\mu$ is 
\[
w(\mu):=\sum_{t_j\in T_{\mu}}q_{j},
\]
where $T_{\mu}\subset T$ is the set of tasks matched by $\mu$.
Given a bipartite graph and a vector $\bb$, the MVbM problem consists in finding a $b$-matching $\mu$ that maximizes  $w(\mu)$. 
Finally, given two sets of edges $\mu_1$ and $\mu_2$, denote with $\mu_1\oplus\mu_2=(\mu_1\backslash\mu_2)\cup(\mu_2\backslash\mu_1)$ their \emph{symmetric difference}.
A path $P = \{ (t_{j_1},a_{i_1}), (a_{i_1},t_{j_2}), (t_{j_2},a_{i_2}), \cdots , (t_{j_L}, a_{i_L}) \}$ in $G$ is a sequence of edges that joins a sequence of vertices. 
We say that $P$ has a length equal to $L$ if it contains $L$ edges.
Given a path $P$ and a $b$-matching $\mu$, $P$ is an \emph{augmenting path} with respect to $\mu$ if every vertex $a_{i_\ell}$, for $\ell=1,\dots,L-1$, is saturated with respect to $\mu$, $a_{i_L}\in A$ is unsaturated with respect to $\mu$, and the edges in the path alternatively do not belong to $\mu$ and belong to $\mu$. 
That is, $(t_{j_\ell},a_{i_\ell})\notin \mu$,
$\ell\in[L]$ and $(a_{i_\ell}, t_{j_{\ell+1}})\in \mu$, $\ell\in[L-1]$, where $[L]$ is the set containing the first $L$ natural numbers, i.e. $[L]:=\{1,2,\dots,L\}$.

{\bf The Algorithm Outline.}
In this paper, we consider two well-known algorithms from a mechanism design perspective.
The first algorithm is presented in \cite{spencer1984node} and its routine consists in defining a sequence of matchings, namely $\{\mu_j\}_{j=0,1,\dots,m}$, of increasing cardinality.
Indeed, the first matching is set as $\mu_0=\emptyset$ and, given $\mu_{j-1}$, $\mu_j$ is defined as follows: if there exists $P_j$ an augmenting path (with respect to $\mu_{j-1}$) that starts from $t_j$, then $\mu_j=\mu_{j-1}\oplus P_j$, otherwise $\mu_j=\mu_{j-1}$. 
In Algorithm \ref{alg:Maxvalue}, we sketch the routine of the algorithm.
In \cite{spencer1984node} it has been shown that Algorithm \ref{alg:Maxvalue} returns a solution to the MVbM problem in $O(NM)$ time, regardless of whether we find the augmenting path using the \emph{Breadth-First Search} (BFS) or the \emph{Depth-First Search} (DFS). 
Both the DFS and the BFS when they traverse the graph in search of an augmenting path implicitly assume that there is an ordering of the agents.
We say that agent $a_i$ has a higher \emph{priority} than agent $a_j$ if $a_i$ is explored before $a_j$, so that agent $a_1$ has the highest priority.
We say that Algorithm \ref{alg:Maxvalue} is endowed with the BFS if we use the BFS to find an augmenting path.
Similarly, we say that the algorithm is endowed with the DFS if we use the DFS to find an augmenting path.
The second algorithm is introduced in \cite{dobrian20192} and \cite{al20222} and is an approximation version of Algorithm \ref{alg:Maxvalue} that searches only among the augmenting paths whose length is $1$.
The authors showed that this approximation algorithm finds a matching whose weight is, in the worst case, half the weight of the MVbM.
Notice that for this approximation algorithm, using the BFS or the DFS does not change the outcome thus we omit which traversing graph method is used.

\begin{algorithm}[tb]
\caption{Algorithm for MVbM}
\label{alg:Maxvalue}
\textbf{Input}: A bipartite graph $G=(A \cup T,E)$; agents' capacities $b_{i}, i=1,\cdots,n$; task weights $q_j, j=1,\cdots,m$.\\
\textbf{Output}:  An MVbM $\mu$.

\begin{algorithmic}[1] 
\STATE $\mu_0\leftarrow\emptyset$;
\STATE Sort $q_j$, $j=1,\cdots,m$, in decreasing order;
\FOR{each $j\in T$}
\IF {there is augmenting path $P$ starting from $j$ with respect to $\mu_{j-1}$}
\STATE $\mu_j=\mu_{j-1}\oplus P$;
\ELSE
\STATE $\mu_{j}=\mu_{j-1}$;
\ENDIF
\ENDFOR
\STATE \textbf{return} $\mu=\mu_{m}$;
\end{algorithmic}
\end{algorithm}

\section{The MVbM Problem with Strategic Agents} 
\label{sec:agentmanipulation}
In this section, we introduce and study a game-theoretical framework for the MVbM problem.
We focus on the case in which the agents are strategic entities able to hide part of their connections to the set of tasks and/or part of their own capacity.

\subsection{The Game-Theoretical framework for agents' manipulation}
\label{subsec:gtfwagent}

Let us now assume that every agent in $A$ reports its private information to a mechanism that, after gathering all the reports, returns a $b$-matching.
In this subsection, we formally define the space of the agents' strategies and the mechanisms we are studying.

{\bf The Strategy Space of the Agents.}
Given a bipartite graph $G=(A\cup T, E)$, a capacity vector $\bb$, a value vector $\qq$, and a $b$-matching $\mu$ over $G$, we define the social welfare achieved by $\mu$ as the total value of the tasks assigned to the agents.
Since the social welfare achieved by $\mu$ is equal to the total weight of the matching, we use $w(\mu)$ to denote it.
We then define the utility of agent $a_i$ as 
\[
w_i(\mu):=\sum_{t_j\in T_{\mu,i}} q_{j},
\]
where $T_{\mu,i}$ is the set of tasks that $\mu$ matches with agent $a_i$.
Notice that $w(\mu)=\sum_{i=1}^{n}w_i(\mu)$.
In this section, we consider two settings depending on what the private information of the agents is: 
\begin{enumerate}[label=(\roman*)]
    \item The private information of each agent consists of the set of edges that connect the agent to the tasks' set. We call this setting Edge Manipulation Setting (EMS).
    \item The private information of each agent consists of the set of edges and its own capacity. We call this setting Edge and Capacity Manipulation Setting (ECMS).
\end{enumerate}
In both cases, we assume that every agent is bounded by its statement, thus every agent is able to report incomplete information, but they are not allowed to report false information.
In EMS, this means that an agent can hide some of the edges that connect it to the tasks, but it cannot report an edge that does not exist.
The set of strategies of agent $a_i$, namely $\mathcal{S}_i$, is, therefore, the set of all the possible non-empty subsets of $T_i$, where $T_i$ is the set containing all the existing edges that connect $a_i$ to the tasks.
In ECMS, this means that an agent can report only a capacity that is lower than its real one and that it cannot report an edge that does not exist.
In this case, the set of strategies of agent $a_i$ is then $\mathcal{S}_i\times [b_i]$, where $b_i$ is the real capacity of the agent and $[b_i]:=\{1,\dots,b_i\}$.
For the sake of simplicity, \textbf{from now on all the results and definitions we introduce are for the EMS}, unless we specify otherwise.
We generalize our results to the ECMS in subsection 3.5.
{\bf The Mechanisms.}
A mechanism for the MVbM problem is a function $\mathbb{M}$ that takes as input the private information of the agents and returns a $\bb$-matching.
We denote with $\mathcal{I}_{\mathbb{M}}$ the set of  possible inputs for $\mathbb{M}$ in the EMS.
In our paper, we consider three mechanisms:
\begin{enumerate}[label=(\roman*)]
    \item $\MB$, which takes in input the edges of the agents and uses Algorithm \ref{alg:Maxvalue} endowed with the BFS to select a matching\footnote{Throughout the paper, we use $\MB$ to denote both the mechanism that takes as input the edges and both the edges and the capacity of every agent. It will be clear from the context which is the input of the mechanism. The same goes for the other two mechanisms.}.
    \item $\MD$, which takes in input the edges of the agents and uses Algorithm \ref{alg:Maxvalue} endowed with the DFS to select a matching.
    \item $\MG$, which takes in input the edges of the agents and uses the approximated version of Algorithm \ref{alg:Maxvalue} to select a matching.
\end{enumerate}
Notice that $\MB$ and $\MD$ are optimal since they are both induced by Algorithm \ref{alg:Maxvalue}.
In Figure \ref{fig:test}, we report an example of truthful input and the different outputs returned by $\MB$, $\MD$, and $\MG$.
Given a mechanism $\mathbb{M}$, every element $I\in \mathcal{I}_{\mathbb{M}}$ is composed by the reports of $n$ self-interested agents, so that $\mathcal{I}_{\mathbb{M}}=\otimes_{i=1}^n\mathcal{I}_{i}$, where $\mathcal{I}_i$ is the set of feasible inputs for agent $a_i$.
We say that a mechanism $\mathbb{M}$ is \textit{strategyproof} (or, equivalently, \textit{truthful})  for the EMS if no agent can get a higher payoff by hiding edges.
More formally, if $I_i$ is the true type of agent $a_i$, it holds true that 
\[w_i(\mathbb{M}(I'_i,I_{-i}))\le w_i(\mathbb{M}(I_i,I_{-i}))
\]
for every $I'_i\in \mathcal{S}_i$.
Another important property for mechanisms is the group strategyproofness.
A mechanism is \textit{group strategyproof} for agent manipulations if no group of agents can collude to hide some of their edges in such a way that
\begin{enumerate}[label=(\roman*)]
    \item the utility obtained by every agent of the group after hiding the edges is greater than or equal to the one they get by reporting truthfully,
    \item at least one agent gets a better payoff after the group hides the edges.
\end{enumerate}

\begin{figure}[t]
\centering
\begin{subfigure}{.5\textwidth}
  \centering
  \includegraphics[width=.45\linewidth]{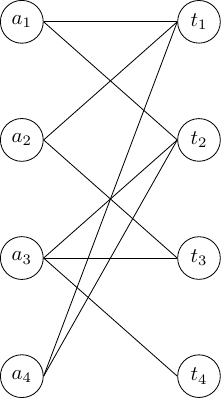}
  \caption{The truthful input.}
  \label{fig:sub1}
\end{subfigure}%
\begin{subfigure}{.5\textwidth}
  \centering
  \includegraphics[width=.45\linewidth]{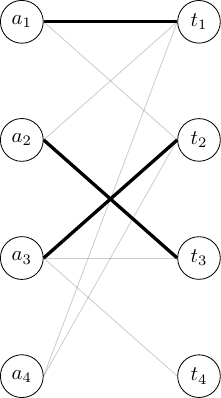}
  \caption{The output of $\MG$.}
  \label{fig:sub2}
\end{subfigure}
\begin{subfigure}{.5\textwidth}
  \centering
  \includegraphics[width=.45\linewidth]{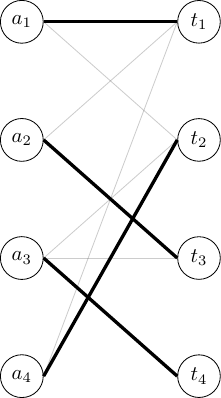}
  \caption{The output of $\MB$.}
  \label{fig:sub12}
\end{subfigure}%
\begin{subfigure}{.5\textwidth}
  \centering
  \includegraphics[width=.45\linewidth]{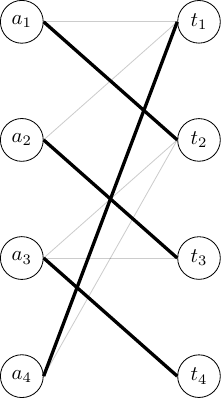}
  \caption{The output of $\MD$.}
  \label{fig:sub22}
\end{subfigure}
\caption{A truthful input and the different outputs provided by $\MB$, $\MD$, and $\MG$. The thicker edges are the one composing the output of the considered mechanism.}
\label{fig:test}
\end{figure}

To evaluate the performances of the mechanisms, we use the Price of Anarchy (PoA), Price of Stability (PoS), and approximation ratio (ar), which we briefly recall in the following.
\textbf{Price of Anarchy.}
The Price of Anarchy (PoA) of mechanism $\mathbb{M}$ is defined as the maximum ratio between the optimal social welfare and the  welfare in the worst Nash Equilibrium, hence
\[
PoA(\mathbb{M}):=\sup_{I\in \mathcal{I}}\frac{w(\mu(I))}{w(\mu_{wNE}(I))},
\]
where $\mu_{wNE}(I)$ is the output of $\mathbb{M}$ when the agents act according to the worst Nash Equilibrium, i.e. the Nash Equilibrium that achieves the worst social welfare.

\textbf{Price of Stability.}
The Price of Stability (PoS) of a mechanism $\mathbb{M}$ is defined as the maximum ratio between the optimal social welfare and the welfare in the best Nash Equilibrium, hence
\[
PoS(\mathbb{M}):=\sup_{I\in \mathcal{I}}\frac{w(\mu(I))}{w(\mu_{bNE}(I))},
\]
where $\mu_{bNE}(I)$ is the output of $\mathbb{M}$ when the agents act according to the best Nash Equilibrium, i.e. the Nash Equilibrium that achieves the maximum social welfare.
Notice that, by definition, we have $PoS(\mathbb{M})\le PoA(\mathbb{M})$.

\textbf{Approximation Ratio.} 
The approximation ratio of a truthful mechanism $\mathbb{M}$ is defined as the maximum ratio between the optimal social welfare and the welfare returned by $\mathbb{M}$.
Hence, we have 
\[
ar(\mathbb{M}):=\sup_{I\in \mathcal{I}}\frac{w(\mu(I))}{w(\mu_{\mathbb{M}}(I))},
\]
where $\mu_{\mathbb{M}}(I)$ is the output of $\mathbb{M}$ when the input is $I$.
%


\subsection{The Truthfulness of the Mechanisms}
\label{sec:game_theoretical_setting}

In this section, we study the truthfulness of the three mechanisms induced by Algorithm \ref{alg:Maxvalue} and its approximation version. 
We show that, although $\MB$ and $\MD$ are optimal, they are not truthful due to an impossibility result.
Furthermore, we show that the manipulability of a mechanism is related to the length of the augmenting paths found during the routine of Algorithm \ref{alg:Maxvalue} and use this characterization to prove that $\MG$ is truthful.

\begin{theorem}
\label{th:no_strategyproof_agent}
There is no deterministic truthful mechanism that always returns an MVbM for agent manipulation.
\end{theorem}

\begin{proof}[Proof of Theorem \ref{th:no_strategyproof_agent}]
We show this using a counterexample.
Consider two agents $a_1$ and $a_2$ and three tasks $t_1$, $t_2$, $t_3$.
The edge set is $E=\{(a_1,t_1), (a_1,t_2),(a_2,t_1), (a_2,t_3)\}$. 
The values of the three tasks are $q_1 = 1, q_2 = 0.1$, and $q_3=0.1$, respectively, while the capacity of both agents is $b_1=b_2=1$. 
It is easy to see that the optimal matching is not unique. 
In particular, both $\{(a_1,t_1)$, $(a_2,t_3)\}$ and $\{(a_1,t_2),(a_2,t_1)\}$ are feasible solutions whose total weight is $1.1$.
Let us assume that the mechanism returns $\{(a_1,t_1)$, $(a_2,t_3)\}$. 
In this case, if agent $a_2$ does not report edge $(a_2,t_3)$, the maximum matching becomes $\{(a_1,t_2),  (a_2,t_1)\}$.
According to this, agent $a_2$'s utility increases from $0.1$ to $1$.
Similarly, if the mechanism returns $\{(a_1,t_2),(a_2,t_1)\}$, agent $a_1$ can manipulate the result by hiding the edge $(a_1,t_2)$.
Therefore, there is no deterministic truthful mechanism that always returns a maximum matching.
\end{proof}

From Theorem \ref{th:no_strategyproof_agent}, we infer that both $\MB$ and $\MD$ are not truthful and, thus, not group strategyproof.
In the following, we characterize a sufficient condition under which agents' best strategy is to report truthfully.
This characterization allows us to deduce that $\MG$ is truthful and to present sets of instances on which $\MB$ and $\MD$ are truthful.

\begin{lemma}
\label{thm:truthfulness_cases}
Let us consider an instance in which Algorithm \ref{alg:Maxvalue} completes its routine using only augmenting paths of length equal to $1$. 
If an agent cannot misreport edges in such a way that the algorithm finds an augmenting path of length greater than $1$ in its implementation, then the agent's best strategy is to report truthfully.
\end{lemma}

\begin{proof}[Proof of Lemma \ref{thm:truthfulness_cases}]

Let us denote with $\mu_k$ the matching found at the $k$-th step.
By hypothesis, Algorithm \ref{alg:Maxvalue} terminates by using only augmenting paths of length equal to $1$, thus the sequence of matching that it finds is monotone increasing. 
That is, $\mu_k\subset \mu_{k+1}$ for every $k$.
Since the sequence $\{\mu_k\}_k$ is increasing, it means that the algorithm is able to define $\mu_{k+1}$ by simply adding an edge connected to $t^{k+1}$ to $\mu_k$. 
Let us assume that an agent hides a set of edges and that the matching sequence found by the algorithm is still monotone increasing. 
Since the sequence of matching found in the manipulated instance is monotone increasing, it means that the final matching is obtained by adding an edge at every iteration. 
By the definition of $\MB$ and $\MD$ and since the length of the augmenting path is $1$, the edge added at step $j$, is taken from the ones that are connected to the task $t_j$.
Therefore, the agent manipulating cannot benefit from the manipulation, since it is only reducing the set of tasks that it can receive. 
\end{proof}

Since $\MG$ uses only augmenting paths of length at most equal to $1$, regardless of the input, Lemma \ref{thm:truthfulness_cases} allows us to conclude that $\MG$ is truthful.
Moreover, due to the results proven in \cite{dobrian20192}, we are able to characterize the approximation ration of $\MG$.

\begin{theorem}
\label{thm:approx_truth}
The mechanism $\MG$ is truthful with respect to agent manipulation. 
Moreover, its approximation ratio is $2$.
\end{theorem}

\begin{proof}[Proof of Theorem \ref{thm:approx_truth}]
The first part of the statement follows directly from Lemma \ref{thm:truthfulness_cases}.
The second part of the statement has been proven in \cite{dobrian20192}.
Indeed, in \cite{dobrian20192}, it has been shown that the weight of the matching returned by the approximation algorithm that uses augmenting paths of length equal to $1$ is, in the worst case, half of the maximal weight.
In their paper, Dobrian et al study the case in which both sides of the bipartite graph have a non-negative value.
However, if we set the value of each agent in $A$ to be $0$, we have that the weight of the matching studied in \cite{dobrian20192} is the same as the social welfare (from the agents' viewpoint) we defined in the present paper.
Similarly, the maximum weight of the matching is the maximum social welfare achievable.
We therefore conclude that the ratio between the social welfare returned by $\MG$ and the maximum one is greater than $0.5$.
Since the approximation ratio is the ratio between the maximum social welfare and the one achieved by $\MG$, we infer that it has to be smaller than $2$.
To prove that the approximation ratio of the mechanism is $2$, consider the following instances.
There are two agents, namely $a_1$ and $a_2$, and two tasks $t_1$ and $t_2$, whose values are $1+\epsilon$ and $1$, respectively, where $\epsilon>0$.
Let $E=\{(a_1,t_1),(a_1,t_2),(a_2,t_1)\}$ be the truthful set of edges.
It is easy to see that $\MG$ returns the matching $\mu_{\MG}=\{(a_1,t_1)\}$, while the optimal one is $\{(a_2,t_1),(a_1,t_2)\}$.
Hence the ratio between the maximum social welfare and the social welfare achieved by $\MG$ on this instance is $\frac{2+\epsilon}{1+\epsilon}$.
By taking the limit for $\epsilon\to 0$, we conclude the proof.
\end{proof}

The approximation ratio achieved by $\MG$ is actually the best possible ratio achievable, as the next result shows.

\begin{theorem}
\label{thm:tight_bound_general}
  There does not exist a truthful and deterministic mechanism for the MVbM problem that achieves an approximation ratio better than $2$.  
\end{theorem}

\begin{proof}[Proof of Theorem \ref{thm:tight_bound_general}]
    Toward a contradiction, assume that there exists a mechanism $\mathbb{M}$ whose approximation ratio is equal to $2-\delta$, where $\delta$ is a positive value.
    Let us now consider the following instance.
    We have two agents, namely $a_1$ and $a_2$, and two tasks, namely $t_1$ and $t_2$.
    The capacity of both agents is set to be equal to $1$.
    The value $q_1$ of $t_1$ is $1+\epsilon$ and the value $q_2$ of $t_2$ is equal to $1$.
    Finally, we assume that, according to their truthful inputs, both agents are connected to both tasks, so that $E=\{(a_1,t_1),(a_1,t_2),(a_2,t_1),(a_2,t_2)\}$.
    It is easy to see that the maximum value that a matching can achieve is $2+\epsilon$. 
    If the mechanism $\mathbb{M}$ does not allocate both the tasks, we have that the value achieved by the mechanism is, at most $1+\epsilon$.
    Therefore, we have that the approximated ratio of $\mathbb{M}$ is at least equal to $\frac{2+\epsilon}{1+\epsilon}=1+\frac{1}{1+\epsilon}$.
    If we take $\epsilon <\frac{\delta}{1-\delta}$, we get that the approximated ratio of $\mathbb{M}$ should be greater than $2-\delta$, which is a contradiction.
    Hence $\mathbb{M}$ allocates both the tasks in the previously described instance.
    Without loss of generality, let us assume that $\mathbb{M}$ allocates $t_1$ to $a_2$ and $t_2$ to $a_1$.
    Let us now consider the instance in which $a_1$ is not connected with the task $t_2$, so that $E'=\{(a_1,t_1),(a_2,t_1),(a_2,t_2)\}$.
    The maximal value that a matching can achieve is still $2+\epsilon$.
    However, since the mechanism is truthful, the agent $a_1$ cannot receive any task.
    Indeed, the only task that $\mathbb{M}$ can assign to $a_1$ is $t_1$, however, if $\mathbb{M}$ assigns $t_1$ to $a_1$, it means that agent $a_1$ can manipulate $\mathbb{M}$ by reporting the set of edges $\{(a_1,t_1)\}$ over $\{(a_1,t_1),(a_1,t_2)\}$ in the instance when the truthful input is $E$.
    Since the first agent cannot receive any tasks from $\mathbb{M}$, we have that the maximum matching value achieved by $\mathbb{M}$ when the input is $E'$ is, at most, $1+\epsilon$, so that the approximation ratio of $\mathbb{M}$ is, at least $\frac{2+\epsilon}{1+\epsilon}$.
    Again, by taking $\epsilon<\frac{\delta}{1-\delta}$, we conclude a contradiction.
    \end{proof}

From Theorem \ref{thm:tight_bound_general}, we then infer that $\MG$ is the best possible deterministic and truthful mechanism for our game theoretical setting.
To conclude, we show that, $\MG$ is also  group strategyproof if all the tasks have different values.

\begin{theorem}
\label{thm:strategyproof}
    If all the tasks in $T$ have different values, then $\MG$ is group strategyproof.
\end{theorem}

\begin{proof}[Proof of Theorem \ref{thm:strategyproof}]
    Toward a contradiction, let us assume that there exists a coalition of agents $C=\{a_{i_1},\dots,a_{i_\ell}\}$ that is able to collude.
    Since hiding edges that are not returned by $\MG$ does not alter the outcome of the mechanism, we assume, without loss of generality, that at least agent $a_{i_1}$, hides one of the edges that are in the matching found by $\MG$ when all the agents report truthfully.
    Let us denote with $t_l$ the task connected to $a_{i_1}$ through the hidden edge.
    After misreporting agent $a_{i_1}$ cannot be allocated with $t_l$.
    Furthermore, due to the routine of $\MG$, each task is allocated independently from the others, hence $a_{i_1}$ cannot be allocated with a better task.
    Finally, since there are no tasks with the same value, $a_{i_1}$'s payoff is necessarily lowered by misreporting, even if in a coalition, which is a contradiction.
    We therefore conclude the proof.
\end{proof}

The condition of Theorem \ref{thm:strategyproof} are tight.
Indeed, as we show in the next example, even if just two tasks have the same value, the mechanism is no longer group strategyproof.

\begin{example}
\label{ex:collusion}
Let us consider the following instance. 
The set of agents is composed of three elements, namely $a_1$, $a_2$, and $a_3$.
The capacity of each agent is set to $1$, so that $b_1=b_2=b_3=1$.
The set of tasks is composed of two elements, namely $t_1$ and $t_2$.
The value of both tasks is equal to $1$.
Finally, let $E=\{(a_1,t_1),(a_1,t_2),(a_2,t_2),(a_3,t_1)\}$ be the truthful input.
Then, $\MG(E)=\{(a_1,t_1),(a_2,t_2)\}$.
However $a_1$ and $a_3$ can collude: if agent $a_1$ hides the edge $(a_1,t_1)$, the $\MG$ returns $\mu'=\{(a_1,t_2),(a_3,t_1)\}$. 
\end{example}

\subsection{Strategy Characterization and Worst-case Equilibrium}

\label{sect:strategy}

In this subsection, we study the Nash Equilibria induced by the mechanisms $\MB$ or $\MD$. 
First, we show that an agent who is not matched to any task when reporting truthfully cannot improve its utility by manipulation.

\begin{lemma}
\label{lemma:nopayoff}
Given a truthful input, if an agent receives a null utility from $\MB$, then its utility cannot be improved by hiding edges.
The same holds for the mechanism $\MD$.
\end{lemma}

\begin{proof}[Proof of Lemma \ref{lemma:nopayoff}]
For the sake of simplicity, we present the proof only for the mechanism $\MB$.
Let us denote with $G$ the truthful graph and with $G'$ the graph manipulated by agent $a_k$. 
Let $\mu$ and $\mu'$ be the matching returned by $\MB$ when $G$ and $G'$ are given as input, respectively. 
By contradiction, suppose $\mu'\neq\mu$. In this case, there exists $t_j$ such that $P_j\neq P_j'$ and $P_r=P_r'$ for any $r<j$, where $P_l$ and $P_l'$ are the augmenting paths returned by the BFS at the $l$-th step.
Since $G'\subset G$ and $\mu_{l-1}=\mu'_{l-1}$, we must have that the last edge of $P_l$ is in $G$ and not in $G'$.
However, since $G'$ is only missing edges connected to $a_k$,  we must have that $P_l$ contains an edge that connects agent $a_k$ to a task. 
In particular, at the $j$-th step, the augmenting path returned by the BFS contains an edge connected to $a_k$, which is not possible since $a_k$ is left unmatched by $\mu$.
\end{proof}

For this reason, starting now, we make the assumption that if an agent can manipulate either $\MB$ or $\MD$, then that agent is assigned at least one task when reporting truthfully.

We ow demonstrate that the first agent processed by either $\MB$ or $\MD$, denoted as $a_1$, is always capable of obtaining its highest possible payoff by misreporting.

\begin{theorem}
\label{thm:max_match_agent_one}
For both $\MB$ and $\MD$, agent $a_1$'s best strategy is to report only the top $b_1$-valued tasks to which it is connected.
\end{theorem}

\begin{proof}[Proof of Theorem \ref{thm:max_match_agent_one}]
Let us denote with $t_{j_1},t_{j_2},\dots,t_{j_{b_1}}$ the top $b_1$ tasks agent $a_1$ is connected to.
For every $t_{j_r}$, at the $j_r$-th step of Algorithm \ref{alg:Maxvalue}, agent $a_1$ will not be saturated; therefore the path $P_{j_r}=\{(t_{j_r},a_1)\}$ is augmenting with respect to matching $\mu_{j_r-1}$.
Moreover, since the BFS searches among the vertices in lexicographical order, the path $P_{j_r}$ will always be the first one being explored and, since it is augmenting, it will be the one returned. 
To conclude, we notice that, after the $j_r$-th iteration, there are no augmenting paths that pass by $a_1$ as all the edges connected to $a_1$ are already in the matching, hence the set of tasks allocated to $a_1$ will not change in later iterations of the algorithm.
By a similar argument, we get to the same conclusion for $\MD$.
\end{proof}

In the next example, we show that the advantage described in Theorem \ref{thm:max_match_agent_one} is only due to the fact that agent $a_1$ is self-aware of its position in the processing process.
Indeed, if the same strategy is played by an agent that has the same truthful input but a different processing priority, the outcome of the manipulation might be detrimental to the agent.

\begin{example}
\label{eq:social_welfarealter}
Let us consider the following instance.
There are $3$ agents, namely $\alpha$, $\beta$, and $\gamma$ whose capacities are $b_\alpha=2$ and $b_\beta=b_\gamma=1$. 
Let us consider a set of $4$ tasks, namely $t_j$ with $j\in [4]$ whose respective values are $q_j=2^{-j}$.
In the truthful input, agent $\alpha$ is connected to all $4$ tasks, while agent $\beta$ is connected only to task $t_1$ and agent $\gamma$ is connected only to task $t_2$.
Let us consider the mechanism $\MB$: the maximum matching for the truthful input allocates $t_1$ to agent $\beta$, $t_2$ to agent $\gamma$, and the other two tasks to agent $\alpha$. 
Let us now assume that the processing order of the agents is $\alpha$, $\beta$, and $\gamma$. That is, $a_1=\alpha$, $a_2=\beta$, and $a_3=\gamma$.
Then, if agent $\alpha$ applies the strategy highlighted in Theorem \ref{thm:max_match_agent_one} and reports only the first two edges, it improves its utility from $q_3+q_4$ to $q_1+q_2$.
However, in a different order of agents in which agent $\alpha$ is the second, i.e., $a_2=\alpha$, if it applies the same strategy, then agent $\alpha$ gets only one of the two tasks (depending on who is the agent going first) while if it goes as the third, it receives no tasks.
We also note that, if agent $\alpha$ is the second, its best strategy is to report the edges connecting it to tasks $t_1$ and $t_3$ if agent $\gamma$ goes first or tasks $t_2$ and $t_3$ if $\beta$ goes first.
Hence, the priority of agent $\alpha$ and the priority of the other two agents determines what the best strategy for agent $\alpha$ is.
Similarly, the same conclusion can be drawn for the mechanism $\MD$.
\end{example}

If the agents' order is changed and the same agent does not have the highest priority, then its best strategy depends on its priority and the agents before it.  
Indeed, the agents' order determines the Nash Equilibrium of both $\MB$ and $\MD$.
For every $i=0,1,\dots,n$, let us define the sets $T^{(i)}$ and $B^{(i)}$ in the following iterative way:
 \begin{enumerate}[label=(\roman*)]
    \item $T^{(0)}=T$ where $T$ is the set of all the tasks and $B^{(0)}=\emptyset$;
    \item $T^{(i)}=T^{(i-1)}\backslash B^{(i-1)}_{i}$, where $B^{(i-1)}_{i}$ is the set containing the top $b_i$-valued tasks among the ones in $T^{(i-1)}$ that agent $i$ is connected to. If there are less than $b_i$ tasks among the ones that the agent can take, it takes all the tasks in $T^{(i-1)}$ that it is connected to.
\end{enumerate}

We then define the \emph{\safepolicypp} (FCFS policy) of agent $a_i$ as $\SP_i=\{(a_i,t_j)\}_{t_j\in B^{(i-1)}_i}$.
Notice that $\SP_i\in \mathcal{S}_i$ for every $i\in [n]$, so that it is a feasible strategy for every agent.

\begin{theorem}
\label{thm:Nash_Equlibrium}
Given an MVbM problem, the \safepolicies\ constitute a Nash Equilibrium that achieves the lowest social welfare in both mechanisms $\MB$ and $\MD$.
Moreover, we have 
\[
\MB(\cup_{a_i\in A}\SP_i)=\MD(\cup_{a_i\in A}\SP_i)=\cup_{a_i\in A}\SP_i.
\]
\end{theorem}

\begin{proof}[Proof of Theorem \ref{thm:Nash_Equlibrium}]
We prove the first part of the theorem in two steps.
First, we prove that the \safepolicies\ constitute a Nash Equilibrium.
Second, we show that the Nash Equilibrium they form is the one with the lowest possible social welfare.
Let us then consider agent $a_i$ and, toward a contradiction, let us assume that, when all the other agents are using their \safepoliciesp, reporting the set of edges $S_{a_i}$ gives $a_i$ a bigger payoff than what it would get from reporting $\SP_i$. 
Let us set $s_i=\min_{(a_i,t_j)\in \SP_{i}} q_j$.
Let us assume that $S_{a_i}$ contains an edge that connects $a_i$ to a task that has a higher value than $s_i$, namely $t_l$.
We now show that, since the other agents are applying their \safepoliciesp, the task $t_l$ is allocated to an agent with higher priority unless the edge already belongs to $\SP_i$. 
Indeed, since $|\SP_j|\le b_j$ for every $a_j\in A$, there cannot be augmenting paths that pass by any of the agents playing their FCFS policy.
Indeed, since the union of the FCFS policies is a $b$-matching, either $(a_i,t_l)\in \SP_i$ or there exists another agent whose FCFS policy connects it to $t_l$.
If $t_l$ is connected to an agent $a_k$, $(a_k,t_l)\in FCFSP_k$, and agent $a_k$'s priority is higher than agent $a_i$'s priority, the final output of the mechanism assigns $t_l$ to $a_k$.
We can then assume that $S_{a_i}$ does not contain edges that connect agent $a_i$ to tasks with values higher than $s_i$ and that do not belong to $\SP_i$.
To conclude, we notice that if $S_{a_i}$ contains an edge connecting agent $a_i$ to a task that has a value lower than $s_i$, then the payoff of agent $a_i$ can only be lowered. 
Indeed, if $|\SP_i|=b_i$, then $a_i$ cannot improve its payoff by reporting edges that connect $a_i$ to tasks that have a value lower than $s_i$.
This follows from the fact that all the other players are using their \safepolicies\ and therefore agent $a_i$ is allocated the set $B^{(i-1)}_i=\{t_j\in T \;\text{s.t.}\; (a_i,t_j)\in \SP_i\}$ if it uses its \safepolicyp.
If $|\SP_i|<b_i$, by definition, it means that there are no tasks that agent $a_i$ can be connected to and that have a value lower than $s_i$.
Therefore $S_{a_i}$ does not contain edges connecting $a_i$ to a task with a value lower than $s_i$ nor edges connecting it with tasks that have a value greater than $s_i$ and that are not included in $\SP_i$.
Since reporting a subset of $\SP_i$ would result in a lower payoff, we deduce that $S_{a_i}=\SP_i$, which is a contradiction since we assumed $\SP_{i}\neq S_{a_i}$.
We now prove that the Nash Equilibrium induced by the \safepolicies\ is one of the worst equilibria.
Toward a contradiction, let us consider another set of strategies, namely $\{S_{a_i}\}_{a_i\in A}$, such that the social welfare achieved by this equilibrium is strictly lower than the one obtained if every agent uses its \safepolicy.
By the definition of social welfare, there must exist at least one agent that, according to the equilibrium defined by $\{S_{a_i}\}_{a_i\in A}$, receives a payoff that is strictly lower than the one it would obtain by using its \safepolicyp.
Let us denote with $a_k$ the first agent that, according to the priority order of the mechanism, receives a lower value.
Agent $a_k$ cannot be the first agent, as it otherwise could apply its \safepolicy\ and get a better payoff.
Then, agent $a_k$ is among the remaining agents and it is not getting any of the tasks that are given to the first agent according to its \safepolicyp, as otherwise, agent $a_1$ could increase its payoff by manipulating and the set of strategies $\{S_i\}_{a_i\in A}$ would not be a Nash Equilibrium.
From a similar argument, we infer that agent $a_k$ cannot be the second agent, as otherwise, it could get a better payoff by using its \safepolicyp.
Moreover, the second agent is allocated the tasks that are granted to it from its \safepolicyp.
Both of which we have already proved cannot be.
By applying the same argument to the other agents, we get a contradiction, since no agent can be agent $a_k$.
We, therefore, conclude that the set of strategies given by the \safepolicies\ is one of the worst Nash Equilibrium.
The last part of the Theorem follows from the fact that $\cup_{a_i\in A}\SP_i$ is itself a $b$-matching, hence it is an $MVbM$ for the edge set $E=\cup_{a_i\in A}\SP_i$. 
\end{proof}

Following the same argument presented in the proof of Theorem \ref{thm:Nash_Equlibrium}, we infer that any set of strategies $\{S_i\}_{a_i\in A}$ for which it holds $\SP_i\subset S_i$ and 
\[
\min_{t_j, \;(a_i,t_j)\in \SP_i} q_j=\min_{t_j, \; (a_i,t_j)\in S_i}q_j
\]
for every $i$, defines a Nash Equilibrium.
Among this class of Nash Equilibria, $\{\SP_i\}_{a_i\in A}$ is the only one for which it holds 
\[
\MB(\cup_{a_i\in A}\SP_i)=\MD(\cup_{a_i\in A}\SP_i)=\cup_{a_i\in A}\SP_i. 
\]
Moreover, this equilibrium achieves the same social welfare of the matching returned by $\MG$.

\begin{theorem}
\label{thm:NE_and_approx}
Given an MVbM problem, let $\mu_E$ be the matching returned by $\MG$.
Then, it holds that $\mu_E=\cup_{a_i\in A}\SP_i$. 
That is, the output of $\MG$ on any given instance is equal to the union of the agents' \safepoliciesp.
Hence, the social welfare achieved by $\MG$ is 
equal to the social welfare achieved by $\MB$ and $\MD$ in one of their worst Nash Equilibria.
\end{theorem}

\begin{proof}[Proof of Theorem \ref{thm:NE_and_approx}]
We prove this Theorem by induction.
First, we prove that $\MG$ allocates $a_1$ with the set of tasks $\SP_1$.
Second, we show that if all the agents $a_1,a_2,\dots,a_k$ receive $\SP_1,\SP_2,\dots,\SP_k$ from $\MG$, then also agent $a_{k+1}$ receives $\SP_{k+1}$.
Let us consider $a_1$. 
Since $\MG$ checks the agents following their orders, $a_1$ is always the first agent checked.
Hence, a task $t_j$ is not allocated to $a_1$ if and only if there are $b_1$ other tasks that have a higher value than $t_j$ and that are all connected to $a_1$.
This means that the set of tasks allocated to $a_1$ is $\SP_1$.
Let us now assume that agents $a_1,a_2,\dots,a_k$ are given the tasks contained in $\SP_1,\SP_2,\dots,\SP_k$ respectively, and let us consider $a_{k+1}$.
We have that $\MG$ allocates a task $t_j$ to $a_{k+1}$ if, at step $j$ of Algorithm \ref{alg:Maxvalue}, all the agents with priorities higher than $a_{k+1}$ that are connected to $t_j$ are already saturated and $a_{k+1}$ is not saturated.
However, by assumption, agents with a higher priority than $a_{k+1}$ are getting the tasks that they would get from their \safepoliciesp.
We, therefore, conclude that the tasks allocated to agent $a_{k+1}$ from $\MG$ consist of a subset of $T^{(k)}$.
From an argument similar to the one used for agent $a_1$, we infer that $a_{k+1}$ receives the top $b_{k+1}$ higher valued tasks among the ones in $T^{(k)}$, which coincides with the set $\SP_{k+1}$.
We conclude the proof by induction.
\end{proof}

Finally, we show that Theorem \ref{thm:NE_and_approx} along with Theorem \ref{thm:approx_truth} allows us to compute the PoA of both $\MB$ and $\MD$.

\begin{theorem}
\label{thm:PoA}
The PoA of $\MB$ and $\MD$ is equal to $2$.
\end{theorem}

\begin{proof}[Proof of Theorem \ref{thm:PoA}]
From Theorem \ref{thm:NE_and_approx}, we know that $\MG$ returns a matching whose social welfare is equal to the social welfare of one of the worst Nash Equilibrium of $\MB$.
Since $\MB$ returns an MVbM, we have
\begin{equation}
    \label{eq:app_socialcompare}
    {\displaystyle PoA(\MB) = \sup_{I\in \mathcal{I}}\frac{w(\mu(I))}{w(\mu_{wNE}(I))}= \sup_{I\in \mathcal{I}} \frac{w(\MB(I))}{w(\MG(I))}.}
\end{equation}
Since the matching found by $\MB$ achieves the maximum social welfare, the last term in equation \eqref{eq:app_socialcompare} is bounded from above by the $a.r.(\MG)$, so that $PoA(\MB) \le ar(\MG)=2$.
To conclude $PoA(\MB)=2$ we show a lower bound of $2$. 
The set of agents is composed of two agents, namely $a_1$ and $a_2$, we assume the agents to be ordered according to the algorithm priority.
The capacity of both agents is equal to $1$. 
The set of tasks contains two tasks, namely $t_1$ and $t_2$, whose values are $1+\epsilon$ and $1$, respectively.
Finally, let us assume that the truthful input is given by $E=\{(a_1,t_1),(a_1,t_2),(a_2,t_1)\}$.
The social welfare is then equal to $2+\epsilon$. 
However, in the worst Nash Equilibrium, the welfare is $1+\epsilon$.
By taking the limit for $\epsilon\to 0$, we conclude that the PoA is equal to $2$.
By a similar argument, we infer $PoA(\MD)=2$.
\end{proof}

The previous bound is tight, indeed, there does not exist a deterministic mechanism for the $MVbM$ problem that has a $PoA$ lower than $2$.

\begin{theorem}
\label{thm:tight_PoA}
    For every deterministic mechanism $\mathbb{M}$, we have $PoA(\mathbb{M})\ge 2$ with respect to agent manipulations.
\end{theorem}

\begin{proof}[Proof of Theorem \ref{thm:tight_PoA}]
    Toward a contradiction, let $\mathbb{M}$ be a deterministic mechanism  whose $PoA$ is lower than $2$.
    Let us consider the following instance. 
    We have two tasks, namely $t_1$ and $t_2$, whose values are $1+\epsilon$ and $1$, respectively.
    We then have two agents, namely $a_1$ and $a_2$ and both have a capacity equal to $1$.
    Let us now consider the instance in which both the agents are only connected to task $t_1$, hence the truthful input is $E=\{(a_1,t_1),(a_2,t_1)\}$.
    Since $PoA(\mathbb{M})$ is finite, we have that $\mathbb{M}$ allocates $t_1$ to one agent.
    Indeed, if no agent receives a task, no one can improve its own payoff by hiding its only edge \footnote{we recall that every agent is bounded by its statement and that every agent has to report at least one edge according to how we defined the strategy set}.
    Hence, the truthful instance is already a Nash Equilibrium.
    Furthermore, since the social welfare of this Nash Equilibrium is $0$, this is also one of the worst Nash Equilibria, thus we find a contradiction since we assumed that $\mathbb{M}$ has a finite PoA.
    Let us then assume that one agent gets $t_1$.
    Without loss of generality, let us assume that $\mathbb{M}$ allocates $t_1$ to $a_1$, the other case is completely symmetric with respect to the one we are about to present.
    Let us now consider the instance whose truthful input is $E=\{(a_1,t_1),(a_1,t_2),(a_2,t_1)\}$.
    If $\mathbb{M}$ allocates $t_1$ to $a_2$, we have that a Nash Equilibrium is obtained when agent $a_1$ hides arc $(a_1,t_2)$.
    Indeed, if agent $a_1$ hides $(a_1,t_2)$, the input of the mechanism is $E=\{(a_1,t_1),(a_2,t_1)\}$ which gives the first task to $a_1$. 
    Since $a_2$ is bounded by its statements and its only alternative is to report no edges, it has no better strategy to play.
    Similarly, since $a_1$ is getting its best possible payoff, it has no better strategy to play.
    We then conclude that when $a_1$ hides the edge $(a_1,t_2)$, we have a Nash Equilibrium.
    Finally, we observe that, by taking $\epsilon$ small enough, we get a contradiction with the assumption $PoA(\mathbb{M})<2$, since the maximum social welfare is $2+\epsilon$, while the social welfare returned by the mechanism in the worst Nash Equilibrium is at most $1+\epsilon$.
    Notice that there might be another Nash Equilibrium in which the social welfare is even lower, however, it suffice to notice that the social welfare achieved in the worst Nash Equilibrium is lower than $1+\epsilon$ to conclude the proof.
    Similarly, if the mechanism does not allocate $t_1$ to $a_2$, the instance is already a Nash Equilibrium.
    Indeed, by the same argument used before, $a_2$ cannot improve its own payoff, since it is getting no tasks.
    If $a_1$ is allocated with the task, it cannot improve its payoff  either, since it is getting the maximum payoff it can get.
    Finally, if $a_1$ is not getting $t_1$, it can hide $(a_1,t_2)$ and return to the instance we considered before.
    Again, by taking $\epsilon$ small enough, we retrieve that the $PoA(\mathbb{M})$ cannot be less than $2$.
\end{proof}

We close the section by studying the PoS of $\MB$ and $\MD$.
We recall that the PoA of every mechanism is greater than its PoS, thus we infer that both $\MB$ and $\MD$ have a PoS at most equal to $2$.
Moreover, since the Nash Equilibrium in the example we used in the proof of Theorem \ref{thm:tight_PoA} is unique, the best and worst Nash Equilibria achieve the same social welfare. In particular, this allows us to prove that the PoS of both $\MB$ and $\MD$ is equal to $2$.
Furthermore, this value is tight for the class of deterministic mechanisms.

\begin{theorem}
\label{thm:PoS}
    The PoS of $\MB$ and $\MD$ is equal to $2$.
    Moreover, there does not exist a deterministic mechanism whose PoS is lower than $2$.
\end{theorem}

\begin{proof}[Proof of Theorem \ref{thm:PoS}]
    Since the Price of Anarchy of a mechanism is always greater than the Price of Stability, we have that $PoS(\MB)\le 2$.
    We now show that $PoS(\MB)= 2$ by showing an instance on which the PoS is equal to $2$.
    Let us consider the example built in the proof of Theorem \ref{thm:PoA}.
    Since there is only one Nash Equilibrium, it is both the worst and best Nash Equilibrium.
    Therefore, following the argument used in the proof of Theorem \ref{thm:PoA}, we conclude $PoS(\MB)\ge 2$, hence $PoS(\MB)=2$.
    From a similar argument, we infer $PoS(\MD)=2$.
    Finally, to prove that the PoS of every deterministic mechanism is greater than $2$, consider the instance build in the proof of Theorem \ref{thm:tight_PoA}.
    Since the Nash Equilibrium of this instance is unique, we conclude that $PoS(\mathbb{M})\ge 2$ for every deterministic and optimal mechanism.
\end{proof}

In particular, $\MB$ and $\MD$ are the best possible mechanisms in terms of PoA and PoS, while $\MG$ is the best possible truthful mechanism in terms of  approximation ratio.
%


\subsection{The Input Space on which \texorpdfstring{$\MB$}{} and \texorpdfstring{$\MD$}{} are truthful.}

\label{subsec:truthinstances}

We now describe some sets of inputs in which $\MB$ and $\MD$ are truthful.
In particular, we consider the three following settings.
In Theorem \ref{thm:truthful_degree}, we consider the case in which there is a shortage of tasks, so that the capacity of each agent exceeds the number of tasks to which it is connected.
In particular, no agent can be saturated.
In Theorem \ref{thm:truth_BFS}, instead, we describe what happens when every task can be contended by another agent.
This means that, for every task $t_j$ that is allocated to an agent, there exists at least another agent that is unsaturated and connected to $t_j$.
Finally, in Theorem \ref{thm:agent_classes}, we study the case in which the private information of the agents can be clustered, that is different agents are connected to the same set of tasks and the same capacity.
Going back to the worker-project example at the beginning of the paper, this means that the connection between the worker and the project depends, for example, on the field of expertise of the worker or their background formation.
In this framework, it is reasonable to assume that there are several different agents that share the same expertise or background.

\begin{theorem}
\label{thm:truthful_degree}
Let us consider the set of inputs such that, according to $E$, all the agents' degrees are less than or equal to their capacity, that is $\sum_{t_j\in T_i}e_{i,j}\le b_i$ for every $i\in [n]$, then $\MB$ and $\MD$ are truthful on this set of inputs.
\end{theorem}

\begin{proof}[Proof of Theorem \ref{thm:truthful_degree}]
By hypothesis, no agent will be saturated during the routine of Algorithm \ref{alg:Maxvalue}. Hence, any augmenting path found during the routine has a length equal to $1$, regardless of the graph traversing method. 
Since hiding edges cannot lead an agent to be saturated, Lemma \ref{thm:truthfulness_cases} allows us to conclude the proof.
\end{proof}

This is the only class of inputs we are considering on which $\MB$ and $\MD$ behave in the same way.
In the other two frameworks, $\MD$ is not truthful, while its counterpart $\MB$ is.
This is due to the fact that BFS searches for the shortest possible augmenting path in its execution.

\begin{theorem}
\label{thm:truth_BFS}
Let $\mu$ be the matching returned by $\MB$.
If for every task $t_j$ there exists an edge $e\notin \mu$ that connects $t_j$ to an unsaturated agent, then no agents can increase their utility by hiding only one edge. 
In particular, if the truthful input is a complete bipartite graph and the vector of the capacities $\bb$ is such that 
\[
m\le \sum_{i=1}^{n} b_i-\max_{i\in [n]}{b_i},
\]
then the best strategy for every agent is to report truthfully.
\end{theorem}

\begin{proof}[Proof of Theorem \ref{thm:truth_BFS}]
Assume, toward a contradiction, that an agent, namely $a_i$ gets a benefit by hiding an edge, namely $e=(a_i,t_j)$.
By hypothesis, we have that every task is connected to an unsaturated agent.
Let $a_k$ be one of the unsaturated agents to which $t_j$ is connected to.
Then BFS will always find an augmenting path whose length is $1$ when it is asked to allocate $t_j$, since there exists the augmenting path $(a_k,t_j)$.
Therefore, by Lemma \ref{thm:truthfulness_cases} we infer a contradiction.
\end{proof}

As the next example shows, Theorem \ref{thm:truth_BFS} cannot be extended to mechanism $\MD$. 

\begin{example}
\label{ex:bip_DFS_vs_BFS}
Let us consider the following instance.
We have three agents, namely $a_i$ for $i=1,2,3$, ordered according to their indexes.
The capacity of each agent is equal to one. 
Let the set of tasks be composed of two tasks $t_1$ and $t_2$, whose values are $1$ and $0.5$, respectively. 
Assume that the edge set of the truthful graph is the complete bipartite one, that is $E:=\{ (a_i,t_j) \}_{i=1,2,3;j=1,2}$. 
For this input $\MB$ returns $\mu_{BFS}=\{(a_1,t_1),(a_2,t_2)\}$, while $\MD$ returns $\mu_{DFS}=\{(a_1,t_2),(a_2,t_1)\}$.
Since agent $1$ does not receive its best pick, by Theorem \ref{thm:max_match_agent_one}, we conclude that the DFS algorithm is susceptible to be manipulated by agent $1$ while, according to Theorem \ref{thm:truth_BFS}, the BFS counterpart is not manipulable by any agent. 
Indeed, Lemma \ref{lemma:nopayoff} ensures us that agent $3$ cannot manipulate $\MB$, since the algorithm leaves it unmatched. 
Similarly, agent $1$ cannot manipulate $\MB$ since it is already receiving its best possible allocation. 
Finally, also $a_2$ cannot manipulate $\MB$, since hiding the only edge used during the implementation of Algorithm \ref{alg:Maxvalue} endowed with the BFS would result in either assigning task $t_2$ to agent $3$ or would not change the allocation at all.
\end{example}

Finally, let $\AA=\{A^{(1)},\dots,A^{(r)}\}$ be a partition of $A$.
We say that $A^{(\ell)}$ is the $\ell$-th class of the agents.
Since $\AA$ is a partition, every agent $a_i\in A$ belongs to only one class.
Finally, let us assume that the capacity $b_i$ and the set of edges $T_i$ of every agent $a_i\in A$ depends only on the class $A^{(\ell)}$ to which $a_i$ belongs, so that $b_i=b^{(\ell)}$ and $T_i=T^{(\ell)}$.
Then, we have that $\MB$ is truthful if every class contains enough agents.

\begin{theorem}
\label{thm:agent_classes}
If $|A^{(\ell)}|>\big\lceil \frac{|T^{(\ell)}|}{b^{(\ell)}} \big\rceil+1$, then no agents belonging to the $\ell$-th class can benefit by misreporting.
\end{theorem}

\begin{proof}[Proof of Theorem \ref{thm:agent_classes}]
First, we prove that if there are $k$ agents that report the same edge sets (but that do not have necessarily the same capacity), namely $a_1$, $a_2$, $\dots$, $a_k$, then, if agent $a_l$ is allocated at least a task, all the agents $a_1$, $a_2$,$\dots$, $a_{l-1}$ are saturated.
Moreover, if agent $a_l$ is unsaturated, then all the agents $a_{l+1}$, $\dots$, $a_k$ do not receive any task.
Toward a contradiction, assume that agent $a_k$ is allocated a task and that $a_{k-1}$ is not saturated. 
Since agent $a_k$ receives a task, it means that during an iteration of the Algorithm, the augmenting path ends in $a_k$ with an edge $(t_j,a_k)$. 
However, since $a_{k-1}$ was not saturated and it has the same edges as $a_k$, the BFS should have explored the edge $(t_k,a_{k-1})$ before, which is a contradiction.
Through a similar argument, it is possible to show the second part of the statement.
We are now ready to prove the claim of the Theorem.
Indeed, since every class of agent has at least $\alpha_i+1$ agents, for every class, there is at least one agent that is left unmatched.
Indeed, every agent from the $i$-th class can be allocated at most $b_i$ tasks, hence, if every agent is allocated at least one task, from the statement we proved before, we infer that at least $\alpha_i$ agents are saturated. 
Since we have
\begin{equation}
    \label{eq:app}
    |T_i|\le \alpha_i b_i, 
\end{equation}
we conclude that at least one agent is left unmatched.
We now show that the routine of Algorithm \ref{alg:Maxvalue} over the truthful input terminates without using augmenting paths of length greater than $1$. 
Assume, toward a contradiction, that the routine uses an augmenting path of length greater than $1$ to allocate a task $t$. Since there exists at least one agent that is unmatched and connected to $t$, this cannot be. 
If we show that no agents can force the algorithm to use an augmenting path whose length is greater than $1$, we conclude the proof using Lemma \ref{thm:truthfulness_cases}.
Toward a contradiction, assume that an agent can force an augmenting path of length greater than $1$. 
This means that, during the implementation of $\MB$ of the manipulated input, there exists a task, namely $t_j\in T_i$, that is connected only to saturated agents at step $j$.
Since there are at least $\alpha_i$ agents in the $i$-th category (the manipulating agent is no longer in this class), we conclude the contradiction from equation \ref{eq:app}.
\end{proof}

Again, Theorem \ref{thm:agent_classes} cannot be extended to $\MD$, as the next example shows.

\begin{example}
\label{ex:class_BFS_vs_DFS}
Let us consider the following instance.
We have three tasks, namely $t_j$ with $j=1,2,3$ and $5$ agents, namely $a_i$ with $i=1,\dots, 5$. 
Finally, let us assume that the value of the task $t_j$ is equal to $q_j=3^{-j}$.
Furthermore, let us assume that there are two classes of agents.
The first class has a capacity equal to $2$ and it is connected to all three tasks. 
The second class of agents also has a capacity equal to $2$ and it is connected to only the last $2$ tasks, which are $t_2$ and $t_3$.
Let us now assume that the agents $a_1$, $a_3$, and $a_4$ belong to the first category, while $a_2$ and $a_5$ belong to the second one.
It is easy to see that this problem satisfies the conditions of Theorem \ref{thm:agent_classes}.
We, therefore, deduce that no agent can manipulate $\MB$.
It is easy to see that the output of $\MB$ is $\mu_{BFS}=\{(a_1,t_1),(a_1,t_2),(a_2,t_3)\}$ while $\MD$ returns $\mu_{DFS}=\{(a_2,t_1),(a_1,t_2),(a_1,t_3)\}$.
Since agent $1$ does not receive the top task to which it is connected, we conclude that agent one can manipulate $\MD$ in its favor. 
On the contrary, we notice that no agent can manipulate $\MB$ since $3$ agents receive no payoff, hence they cannot manipulate and the first agent gets its best possible payoff, hence it cannot improve it.
Finally, we notice that if agent $a_2$ manipulates, the only change that it could achieve is that task $t_3$ is allocated to agent $a_3$ rather than to $a_2$.
\end{example}

\subsection{Agents Manipulating their Edges and Capacity}
\label{subsect:ECMS}

In this section, we extend our study on the truthfulness of $\MB$, $\MD$, and $\MG$ to the ECMS, hence the agents self-report their capacity along with their edges.\footnote{With a slighty abuse of notation, we use $\MB$, $\MD$, and $\MG$ to denote the mechanisms obtained from Algorithm \ref{alg:Maxvalue} and its approximation version when applied to the ECMS case.}
As for the EMS, we assume the agents to be bounded by their statements, thus they can manipulate only by hiding edges or by reporting a lower capacity than their real one.
In this setting, a mechanism $\mathbb{M}$ is truthful if, for every $i\in [n]$, it holds $w_i((I'_i,b_i'),J_{-i})\le w_i((I_i,b_i),J_{-i})$, for every $(I_i',b'_i)\in\mathcal{S}_i\times [b_i]$, where $J_{-i}$ are the reports of the other agents.
Once we fix the set of strategies of each agent, we can define the PoA, PoS, and approximation ratio of a mechanism $\mathbb{M}$ as for the EMS by carefully changing the set of strategies to fit the ECMS case.
As we show, neither the truthfulness nor the efficiency guarantees of $\MB$, $\MD$, and $\MG$ change from EMS to ECMS. 
Furthermore, all the bounds are still tight.

\begin{theorem}
\label{thm:PoA_PoS_edge_cap}
    In the ECMS, $\MB$ and $\MD$ are both not truthful.
    The PoA and the PoS of both $\MB$ and $\MD$ are equal to $2$.
    Moreover, these bounds are tight, hence there is no other deterministic mechanism whose PoA or PoS is lower.
\end{theorem}

\begin{proof}[Proof of Theorem \ref{thm:PoA_PoS_edge_cap}]

    First, we observe that, since the strategy set of every agent in the ECMS is larger than what it is in the EMS, the mechanisms cannot be truthful in the ECMS.
    Indeed, if the mechanisms are not truthful for the EMS, they cannot be truthful for the ECMS.
    To prove that $PoA(\MB)=PoA(\MD)=2$ it suffices to show that the strategies $\{(FCFS_i,b_i)\}_{i\in [n]}$ forms a Nash Equilibrium for both $\MB$ and $\MD$.
    Without loss of generality, we focus on $\MB$, since the same conclusion can be drawn for $\MD$.
    Toward a contradiction, let us assume that there exists an agent, namely $a_k$, whose best strategy is to play $(S_k,c_k)$ when all the other agents play $(FCFS_i,b_i)$.
    Since $c_k\le b_k$, agent $a_k$ gets at most the same number of tasks that it would get in the truthful case.
    Furthermore, since the tasks that have a higher value than the one $a_k$ would get by playing $(FCFS_k,b_k)$ cannot be de-allocated from the agents already getting them (since every other player is playing $(FCFS_i,b_i)$), agent $a_k$ is getting fewer tasks that have a lower value than the ones it would get by playing $(FCFS_k,b_k)$, which is a contradiction.
    We deduce that this Nash Equilibrium is one of the worst ones, by using the same argument used in the proof of Theorem \ref{thm:Nash_Equlibrium}.
    Indeed, there cannot exist a Nash Equilibrium in which one agent gets a payoff lower than the one it would get by playing $(FCFS_i,b_i)$.
    We therefore conclude that $\{(FCFS_i,b_i)\}_{i\in [n]}$ is one of the worst Nash Equilibria, thus the PoA of $\MB$ (and $\MD$) is equal to $2$.
    To conclude that $PoS(\MB)=PoS(\MD)=2$, consider the instance described in the proof of Theorem \ref{thm:PoS}.
    Indeed, since in this instance, there is only one possible Nash Equilibrium, it is both the worst and best equilibrium possible.
    Moreover, since the PoS of this instance is equal to $2$ and $PoS(\MB)\le PoA(\MB)=2$, we conclude the proof.
    Similarly, we have $PoS(\MD)=2$.
    To show that all the bounds are tight, it suffice to notice that, according to our framework, an agent whose capacity is equal to $1$ cannot manipulate the mechanism by reporting a null capacity.
    Indeed, an agent that reports a null capacity would automatically be excluded from the allocation procedure.
    Since in the instances we used to prove Theorem \ref{thm:tight_PoA} and \ref{thm:PoS} all the agents have a capacity equal to $1$, we can use them to prove the tightness of the bounds in the ECMS.
\end{proof}

Similarly, $\MG$ is still truthful in the ECMS, and its approximation ratio is unchanged. 

\begin{theorem}
\label{thm:last_mg_truth_cap}
In the ECMS, $\MG$ is truthful, its approximation ratio is equal to $2$.
Moreover, there is no deterministic truthful mechanism with a better approximation ratio.
Finally, if the tasks have different values, $\MG$ is group strategyproof.
\end{theorem}

\begin{proof}[Proof of Theorem \ref{thm:last_mg_truth_cap}]
To show the truthfulness, it suffices to prove that if an agent misreports only its capacity, it cannot get a higher payoff.
Indeed, if for any given set of edges reported by an agent its best strategy is to report the maximum capacity, we conclude that its best strategy, in the ECMS, is to report all the edges and the maximum capacity.
Since every agent is bounded by its statements, this implies that $\MG$ is truthful.
Let us assume that the capacity of agent $a_i$ is $b_i$ and that it gets a higher payoff by reporting $b_i'<b_i$ as its capacity.
Let us denote with $\mu_i$ the set of tasks allocated to $a_i$ from $\MG$ when it reports truthfully and let us denote with $\mu_i'$ the set of tasks allocated to $a_i$ when it reports $b_i'$ as its capacity.
Toward a contradiction, let $t_j\in \mu_i'$ such that $t_j\notin \mu_i$.
Since $t_j$ is allocated to agent $a_i$ by $\MG$, it means that agent $a_i$ is unsaturated at the $j$-th iteration of Algorithm \ref{alg:Maxvalue} and since $b_i'<b_i$, we conclude that agent $a_i$ is unsaturated at the $j$-th step also when it reports $b_i$ as its capacity.
Therefore, if task $t_j$ is allocated to agent $a_i$ when it reports $b_i'$ as its capacity, the same holds when $a_i$ reports a capacity of $b_i$.
Since every task has a positive value and $\mu'_i\subset \mu_i$, we get a contradiction, therefore the mechanism cannot be manipulated by only misreporting the capacity.
Since the mechanism $\MG$ cannot be manipulated by hiding edges or by reporting a lower capacity, we conclude that $\MG$ has to be truthful when the agents report both quantities.
Finally, since $\MG$ is truthful in the ECMS and the agents are bounded by their statements, its approximation ratio is the same as the one of $\MG$ in EMS (since there is no difference in having the capacities publicly known and reporting them), hence $ar(\MG)=2$.
Again, to prove that the approximation ratio guarantee is tight, it suffice to consider the instance used to prove Theorem \ref{thm:tight_bound_general}.
The group strategyproofness follows by the same argument used in the proof of Theorem \ref{thm:strategyproof} and by observing that reporting a lower capacity cannot increase the payoff of the agents.
\end{proof}

\section{The Tasks Manipulation Case}
\label{sec:tasksmanipulation}
We now consider the case in which the tasks' side is behaving strategically.
As for the agents' case, we consider two frameworks.
First, we consider the case in which the tasks' private information consists of only its adjacent edges.
In the second one, the tasks' private information consists of their adjacent edges and their values.
As we will see, all the mechanisms considered so far are truthful (albeit not group strategyproof).
In particular, when the tasks' side behaves strategically, the optimality of a mechanism does not exclude its truthfulness.

\subsection{The Game-Theoretical framework for tasks' manipulation}

First, we formalize the tasks' manipulating framework by adapting the definitions introduced in subsection \ref{subsec:gtfwagent} to this framework.

{\bf The Strategy Space of the Tasks.}
Given a bipartite graph $G=(A\cup T, E)$, a capacity vector $\bb$, a value vector $\qq$, and a $\bb$-matching $\mu$ over $G$.
We define the social welfare achieved by $\mu$ as the total number of tasks assigned to agents.
Therefore, in this setting, the social welfare achieved by $\mu$ is equal to the total number of edges in the matching, we use $v(\mu)$ to denote it.
We then define the utility of task $t_j$ as $v_j(\mu):=1$ if $t_j$ is assigned to any agent and $0$ otherwise, so that $v(\mu)=\sum_{j=1}^{n}v_j(\mu)$.
We focus on two settings: 
\begin{itemize}
    \item The private information of each task consists of the set of edges that connect it to the agent. As for the agents' framework, we call this setting Edge Manipulation Setting (EMS).
    \item The private information of each task consists of the set of edges and its own value. We call this setting Edge and Value Manipulation Setting (EVMS).
\end{itemize}
As for the agents' case, we assume that every task is bounded by its statement, which means that their reports can be incomplete, but they are not allowed to report false information.
For the EMS, this means that a task can hide some of the edges that connect it to the agents, but it cannot report an edge that does not exist.
The set of strategies of task $t_j$, namely $\mathcal{R}_j$, is therefore the set of all the possible non-empty subsets of $G_j$, where $G_j:=\{e\in E\;\text{s.t.}\; e=(a_i,t_j) \;\text{for some} \; a_i\in A\}$ is the set containing all the existing edges that connect $t_j$ to agents.
In EVMS, it means that a task can report only a value that is lower than its real one and that it cannot report an edge that does not exist.
In this case, the set of strategies of task $t_j$ is $\mathcal{R}_j\times (0,q_j]$, where $q_j$ is the real value of $t_j$.
{\bf The Mechanisms.}
A mechanism for the MVbM problem is a function $\mathbb{M}$ that takes as input the private information of the tasks and returns a $\bb$-matching.
For the sake of simplicity, we restrict our description to the EMS case, since the definitions for the EVMS case are similar.
With a slight abuse of notation, we denote with $\MB$ and $\MD$, the mechanisms induced by Algorithm \ref{alg:Maxvalue} endowed with the respective graph traversing method.
Similarly, we denote with $\MG$ the mechanism induced by the approximated version of Algorithm \ref{alg:Maxvalue}.

Given a mechanism $\mathbb{M}$, every element $I\in \mathcal{I}_{\mathbb{M}}$ is composed by the reports of $n$ self-interested tasks, so that $\mathcal{I}_{\mathbb{M}}=\otimes_{j=1}^m\mathcal{I}_{j}$, where $\mathcal{I}_j$ is the set of feasible inputs for task $t_j$.
We say that a mechanism $\mathbb{M}$ is truthful if no task can get a higher payoff by hiding some of the edges.
More formally, if $I_j$ is the true type of task $t_j$, the mechanism is truthful if it holds true that $v_j(\mathbb{M}(I'_j,I_{-j}))\le v_j(\mathbb{M}(I_j,I_{-j}))$ for every $I'_j\in \mathcal{R}_j$.
Notice that, in this framework, this is equivalent to saying that an unmatched task cannot be matched if it hides some edges.
Finally, a mechanism is group strategyproof for task manipulations if no group of tasks can collude to hide some of their edges in such a way that
\begin{enumerate*}[label=(\roman*)]
    \item all the tasks of the group that were matched in the truthful instance are still matched after the manipulation,
    \item at least one task that was unmatched in the truthful input gets matched to an agent.
\end{enumerate*}

\subsection{The Properties of the Mechanisms}

We now study whether $\MB$, $\MD$, and $\MG$ are group strategyproof when the tasks are behaving strategically and are bound by their statements.
The following theorem implies that any mechanism that returns an MVbM is not group strategyproof for task manipulation.

\begin{theorem}
\label{th:nostrgpoff}
There is no group strategyproof mechanism that always returns an MVbM for task manipulation.
\end{theorem}

\begin{proof}[Proof of Theorem \ref{th:nostrgpoff}]
We prove this statement using a counterexample. 
Consider two agents $a_1$ and $a_2$, three tasks $t_1$, $t_2$, $t_3$, and let the edge set be $E=\{(a_1,t_1),(a_1,t_3),(a_2,t_1),(a_2,t_2)\}$.
The values of the three tasks are $q_1 = 1, q_2 = 0.9$, and $q_3=0.1$, respectively. 
The capacity of both the agents is equal to $1$, so that $b_1=b_2=1$. 
It is easy to see that the only MVbM is $\{(a_1,t_1)$, $(a_2,t_2)\}$ with a total weight of $1.9$.
However, tasks $t_1$ and $t_3$ can collude. 
In fact, if $t_1$ hides the edge $(a_1,t_1)$, the optimal matching becomes $\{(a_1,t_3)$, $(a_2,t_1)\}$ with the total weight $1.1$. 
Therefore, a mechanism that always returns a maximum weight matching is not group strategyproof. 
\end{proof}

From Theorem \ref{th:nostrgpoff} we immediately infer that both $\MB$ and $\MD$ are not group strategyproof.
Through a counterexample, we are able to prove that also $\MG$ is not group strategyproof.

\begin{theorem}
\label{th:MGnotstrgpoff}
The mechanism $\MG$ is not group strategyproof for task manipulation.
\end{theorem}

\begin{proof}[Proof of Theorem \ref{th:MGnotstrgpoff}]
Let us consider the following instance.
We have three tasks, namely $t_1$, $t_2$, and $t_3$, whose values are $q_1=1$, $q_2=0.5$, and $q_3=0.3$, respectively, and two agents, namely $a_1$ and $a_2$.
Let $E=\{(a_1,t_1),(a_1,t_3),(a_2,t_1),(a_2,t_2)\}$ be the truthful input.
It is easy to check that $\MG$ returns the matching $\{(a_1,t_1),(a_2,t_2)\}$.
However task $t_1$ and $t_3$ can collude; indeed, if $t_1$ hides the edge $(a_1,t_1)$, $\MG$ returns the matching $\{(a_2,t_1),(a_1,t_3)\}$.
\end{proof}

Albeit $\MB$, $\MD$, and $\MG$ are not group strategyproof, they are truthful for task manipulation.
In particular, this means that both $\MB$ and $\MD$ are truthful and optimal with respect to tasks manipulation.

\begin{theorem}
\label{thm:groupsp}
$\MB$, $\MD$, and $\MG$ are truthful for task manipulation.
Moreover, we have that $ar(\MB)= ar(\MD) = 1$ and $ar(\MG)=2$.
\end{theorem}

\begin{proof}[Proof of Theorem \ref{thm:groupsp}]
To show that $\MB$ is truthful, it suffices to prove that an unmatched task cannot become matched by hiding the edges that connect it with agents in $\MB$.
Suppose that task $t_j$ is not matched by the matching $\mu$ returned by $\MB$.
It is well known that, if $t_j$ is unmatched by $\mu$, it means that it was not matched during the $j$-th loop of Algorithm \ref{alg:Maxvalue}.
If $t_j$ hides some of its edges, there is still no augmenting path starting from $t_j$ since the generated tree is a sub-tree of the one generated during the truthful iteration. 
We then conclude that $t_j$ cannot be matched by misreporting and, hence, $\MB$ is truthful.
By the same argument, we infer that $\MD$ and $\MG$ are both truthful.
Since both $\MB$ and $\MD$ return an MVbM, we have $ar(\MB)= ar(\MD) = 1$.
Lastly, we show that $ar(\MG)=2$.
Let us consider an MVbM problem, namely $\mathcal{P}:=(A\cup T, E, \textbf{b}, \textbf{q})$.
Without loss of generality, we assume that the tasks in $T$ are ordered in a non-increasing fashion with respect to their value, that is $q_i\le q_j$ if $j\le i$.
Let us then consider the auxiliary MVbM problem $\mathcal{P}':=(A'\cup T', E', \textbf{b}', \textbf{q}')$, for which $A'=A$, $T'=T$, $E'=E$, and $\textbf{b}'=\textbf{b}$, so that the only parameter of $\mathcal{P}'$ that may differ from $\mathcal{P}$ is the vector $\textbf{q}'$.
The vector $\textbf{q}'$ is such that $q'_i=1$ for every $i=1,\dots,m$.
Since all the tasks have the same value in $\mathcal{P}'$, we can assume that the tasks are ordered as in $\mathcal{P}$ while it still holds that $q_j'\le q_i'$ whenever $i\le j$.
Since the set of tasks is ordered in the same way for both problems, $\MG$ returns the same matching, namely $\mu_g$, in both cases.
The weight of $\mu_g$ for the problem $\mathcal{P}'$ is equal to the tasks' social welfare of $\mu_g$ for problem $\mathcal{P}$.
Similarly, the weight of an MVbM solution to problem $\mathcal{P}'$ is equal to the maximal tasks' social welfare for problem $\mathcal{P}$.
Finally, we recall that the algorithm defining $\MG$ returns a matching whose weight is at least half of the weight of an MVbM solution (see \cite{dobrian20192}).
Since this argument holds for every MVbM problem, we infer that
\begin{equation}
\label{eq:ar_MG}
    ar(\MG) = \max_{I\in \mathcal{I}} \frac{v(\mu(I))}{v(\MG(I))} \le 2.
\end{equation}
To conclude the thesis, we show that the bound in \eqref{eq:ar_MG} is tight.
Let us consider the following instance.
There are two agents $a_1$ and $a_2$, whose capacities are both equal to $1$, and two tasks $t_1$ and $t_2$, whose values are $1$ and $0.5$, respectively.
Let $E=\{(a_1,t_1),(a_1,t_2),(a_2,t_1)\}$ be the truthful set of edges.
It is easy to see that $\MG(E)=\{(a_1,t_1)\}$, while the MVbM is $\{(a_2,t_1),(a_1,t_2)\}$.
Hence the ratio between the maximum tasks' social welfare and the tasks' social welfare achieved by $\MG$ is $2$, which concludes the proof.
\end{proof}

\subsection{Tasks Manipulating their Value.}

We now consider the case in which the tasks are allowed to self-report their value.
In this case, a mechanism $\mathbb{M}$ is truthful if it satisfies the condition that $v_j(\mathbb{M}((I'_j, q_j'), J_{-j})) \leq v_j(\mathbb{M}((I_j, q_j), J_{-j}))$ for every $(I'_j, q'_j) \in \mathcal{R}_j \times (0, q_j]$, where $J_{-i}$ is the matrix containing the reports of the other tasks.
To demonstrate that the mechanisms are truthful in the EVMS, we initially establish that all of them maintain truthfulness when the set of edges is publicly known, and the only private information pertains to the value of each task.

\begin{lemma}
\label{thm:truthfulaugmented}
If the tasks are bounded by their statements, then both $\MB$ and $\MD$ are truthful if the only input that the tasks have to provide is their value.
\end{lemma}

\begin{proof}[Proof of Lemma \ref{thm:truthfulaugmented}]
We show that if a task, namely $t_j$, is not matched by Algorithm \ref{alg:Maxvalue} when it reports $q_j$, it cannot be matched by reporting a value $q_j'$, where  $q'_j<q_j$.
Let us denote with $\mu$ and $\mu'$ the matchings obtained when $t_j$ reports value $q_j$ and $q'_j$, respectively.
Toward a contradiction, assume that $t_j$ is not matched by $\mu$, but, according to $\mu'$ it is matched with an agent, namely $a_i$.
Since the set of edges does not change, and both $\mu$ and $\mu'$ are feasible $b$-matchings for both the instances (the one in which task $t_j$ reports $q_j$ and the one in which it reports $q_j'$), we must have
\begin{align*}
    w(\mu)\le w(\mu')&=q'_j+\sum_{(i,j)\in \mu\backslash\{(a,t_j)\}}q_j < q_j+\sum_{(i,j)\in \mu\backslash\{(a,t_j)\}}q_j.
\end{align*}
This concludes the contradiction since $\mu$ is an MVbM.
\end{proof}

By using a different argument, we prove that the same result holds for $\MG$.

\begin{lemma}
\label{thm:MAP_task_val_manipulation}
If the tasks are bounded by their statements, then $\MG$ is truthful if the only input that the tasks have to provide is their value.
\end{lemma}

\begin{proof}[Proof of Lemma \ref{thm:MAP_task_val_manipulation}]
Toward a contradiction, let us assume that there exists a task $t_j$ whose real value is $q_j$ and that, according to $\MG$, is not matched when it reports truthfully but it gets matched if it reports a value $q_j'$ such that $q_j'<q_j$.
Let us denote with $j$ the priority position $t_j$ when it reports the value $q_j$ and with $j'$ the priority position that $t_j$ gets when it reports the value $q_j'$.
By hypothesis, we have $j\le j'$.
According to mechanism $\MG$, task $t_j$ is matched only if at the $j$-th step of Algorithm \ref{alg:Maxvalue} $t_j$ is connected to an unsaturated agent.
Since $j<j'$, if an agent is unsaturated at $j'$-th step and is also unsaturated at the $j$-th step, we conclude that $t_j$ cannot be matched when its priority position is $j'$ if it was not matched with priority position $j$, which is a contradiction.
\end{proof}

Finally, we show that if we allow the tasks to report both their edge set and value, all the mechanisms $\MB$, $\MD$, and $\MG$, are truthful as long as the tasks are bounded by their statements.
Thus, also in this extended case, both $\MB$ and $\MD$ are truthful and optimal.

\begin{theorem}
\label{thm:blablabla}
If the tasks are allowed to report both their value and their edges, the mechanisms $\MB$, $\MD$, and $\MG$ are truthful as long as the tasks are bounded by their statements.
Moreover, we have $ar(\MB)=ar(\MD)=1$ and $ar(\MG)=2$.
\end{theorem}
\begin{proof}[Proof of Theorem \ref{thm:blablabla}]
Indeed, let us consider a task, namely $t_j$, whose truthful input is $(q_j,T_j)$, where $q_j$ is its value and $T_j$ is the set of edges it is connected to.
Toward a contradiction, let us assume that $t_j$ is not matched if it reports truthfully, but gets matched if it manipulates and reports $(q_j',T_j')$.
Since the task is bounded by its statement, we have both $q_j'\le q_j$ and $T_j'\subset T_j$.
Let us now consider the matching returned by the mechanism when the task reports $(q_j',T_j)$ (which means, $t_j$ has only manipulated its value and not $T_j$).
According to Lemma \ref{thm:truthfulaugmented}, if $t_j$ could not get matched by reporting its true value, then it could not be matched by reporting a lower one.
We, therefore, conclude that $t_j$ is not matched when it reports $(q_j',T_j)$.
To conclude, we observe that if the task is not matched when it reports $(q_j'.T_j)$, by Theorem \ref{thm:groupsp}, it cannot be matched when it reports $(q_j',T_j')$ with $T_j'\subset T_j$, which contradicts the initial assumption.
Since both $\MB$ and $\MD$ are optimal, we get $ar(\MB)=ar(\MD)=1$.
By the same argument used in the proof of Theorem \ref{thm:approx_truth}, we conclude $ar(\MG)=2$. 
Indeed, since the agents are bounded by their statements and the mechanisms is truthful, there is no difference between letting the values of the tasks be publicly known and letting the tasks self-report their value.
\end{proof}

\section{Experimental Evaluation}

\label{sec:exp_res}

In this section, we conduct several numerical tests to evaluate the manipulability of both $\MB$ and $\MD$.
Since $\MB$ and $\MD$ are truthful when tasks behave strategically, our analysis pertains to scenarios where agents are the strategic entities.
Furthermore, we narrow our focus to the EMS, as we have proven that the ECMS yields no meaningful differences from the EMS.
Our experiments aim to achieve three primary objectives:
\begin{enumerate*}[label=(\roman*)]
    \item First, we compare the susceptibility of $\MB$ and $\MD$ to manipulation by the first agent.
    \item Second, we focus specifically on $\MB$ and evaluate its vulnerability to manipulation by every agent.
    \item Third, we introduce a randomized version of $\MB$ in which the order of agents is determined by running a lottery before executing the routine of the $\MB$. We refer to this modified mechanism as $\MBR$.
    We then compare the performances of $\MB$ and $\MBR$ to investigate whether randomizing the agents' order has a positive impact on the truthfulness of $\MB$.
\end{enumerate*}

\subsection{Dataset Generation}

We now detail the parameters that characterize how the instances are randomly generated and introduce the class of manipulations we consider.

{\bf The truthful input.} 
Given the number of agents $n$ and the number of tasks $m$, we generate a random truthful instance $G=([n]\cup [m], E, \bb, \qq)$ as follows: 
\begin{itemize}
    \item the capacity $b_i$ of the agent $a_i$ is an integer number between sampled from a uniform probability distribution on the set of natural numbers between $\underline{b}$ and $\overline{b}$, i.e. over the set $\{\underline{b},\underline{b}+1,\dots,\overline{b}-1,\overline{b}\}$;
    \item in all the tests, the value $q_j$ of task $t_j$ is sampled from a $Y=\max\{Z,0\}$, where $Z\sim\mathcal{N}(3,0.77)$.
    \item every possible edge $(a_i,t_j)\in [n]\times [m]$ has a probability $p$ to be in $E$. We will use the term \textit{connection probability} to denote $p$.
\end{itemize}
We consider as truthful every instance generated this way.

{\bf Agent manipulation.} 
Given an agent $a_i$, let $T_i$ be the set of edges connected to $a_i$ so that the set of $a_i$'s possible strategies is the set of subsets of $T_i$.
As a consequence, every agent has $2^{|T_i|}$ possible ways to misreport, which is unfeasible from a computational point of view.
For this reason, we focus only on the case in which the agent hides its connections to the tasks with the lowest value.
Depending on the experimental setting, we describe these strategies using two different parameters, the \textit{cutoff threshold} $T$ and the \textit{number of hidden tasks} $k$.
Given an agent $a_i$ and a cutoff threshold $T$, the $T$-level manipulation of agent $a_i$ consists of hiding all its connections with tasks with a value lower than $T$.
Given an agent $a_i$ and a number of hidden tasks $k$, the $k$-th order manipulation of agent $a_i$ consists of hiding its connections with the $k$ tasks with the lowest value.

\subsection{Comparison between \texorpdfstring{$\MB$}{MB} and \texorpdfstring{$\MD$}{MD}.}

The test we run is as follows: given the number of agents $n$, the number of tasks $m$, and the connection probability $p$, we apply both $\MB$ and $\MD$ on $250$ randomly generated instances.
As a metric to compare the manipulability of the two mechanisms, we consider the ratio between the payoff that the first agent gets by $\MB$ ($\MD$, respectively) when it behaves truthfully against the payoff it gets from $\MB$ ($\MD$) when it follows the strategy described in Theorem \ref{thm:max_match_agent_one}.
By definition, when the ratio is equal to $1$, the agent is not able to manipulate.
In Table \ref{tab_comp:1}, we report the average ratio between the payoff that the first agent gets from $\MB$ ($\MD$, respectively) and the maximal payoff it could get when the number of tasks is $m=30, 50, 70$, the number of agents is $n=20,40,60,80$, and $p=0.4,0.6,0.8$.
For the sake of simplicity, we set the capacity of every agent equal to $3$, so that $\underline{b}=\overline{b}=3$.
For every set of parameters, we compute the average by running the experiment over $250$ randomly generated instances.
From our results, we observe that the payoff returned by $\MB$ to the first agent is bigger than the one returned by $\MD$ except on two sets of instances ($n=20$, $m=70$, and $p=0.4,0.6$).
Additionally, in the overwhelming majority of cases, the ratio between the first agent's payoff returned by $\MB$ and what it would receive through manipulation is approximately $1$. Consequently, for these parameter settings, the first agent cannot gain an advantage by manipulating $\MB$ on any of the $250$ randomly generated instances.
It is interesting to notice that the same does not hold for $\MD$, as its average ratio is consistently around $90\%$.
We also observe that $\MB$ gets less manipulable by the first agent as the parameters $n$ and $p$ grow.
Again, we observe a distinct behaviour from $\MD$, for which the first agent's ability to manipulate $\MD$ seems to remain constant for every value of $n$, $m$, and $p$.
To further enhance our observations, we report in Table \ref{tab_comp:2} and \ref{tab_comp:2_cont} the maximum and minimum loss percentages of the first agent among the $250$ iterations we run.
We notice that, in the vast majority of cases, the maximum loss percentage of the $\MB$ is equal to $0$, which means that the first agent was unable to manipulate $\MB$ on all of the $250$ randomly generated instances.
On the contrary, the maximum loss percentage of the $\MD$ is always greater than zero.
Similarly, we observe that the minimum loss percentage of the first agent from $\MD$ is equal to $0$ only in two cases.
This means that, for all but these two parameter choices, on all the $250$ instances we randomly generated according to those parameters, $\MD$ was always manipulable by the first agent.

\begin{table*}[t]

\small

\centering

\begin{tabular}{l| c|c| c|c| c|c |c|c}
    \toprule[2pt]

    \multicolumn{9}{c}{$m=30$}\\
    \midrule
    \multirow{2}{*}{} &
      \multicolumn{2}{c|}{$n=20$} &
      \multicolumn{2}{c|}{$n=40$} &
      \multicolumn{2}{c|}{$n=60$} & \multicolumn{2}{c}{$n=80$}  \\
      \hhline{~|-|-|-|-|-|-|-|-}
      & {$\MB$} & {$\MD$} & {$\MB$} & {$\MD$} & {$\MB$} & {$\MD$} & {$\MB$} & {$\MD$} \\
      \hline
    $p=0.4$ & 0.99 & 0.83 & 1.0 & 0.83 &  1.0 & 0.83 & 1.0 & 0.83 \\
    $p=0.6$ & 1.0 & 0.87 & 1.0 & 0.86 & 1.0 & 0.86 & 1.0 & 0.87 \\
    $p=0.8$ & 1.0 & 0.89 & 1.0 & 0.89 & 1.0 & 0.89 & 1.0 & 0.89  \\
    
    \midrule[1.5pt]
    
    \multicolumn{9}{c}{$m=50$}\\
    \midrule
  &   \multicolumn{2}{c|}{$n=20$} &
      \multicolumn{2}{c|}{$n=40$} &
      \multicolumn{2}{c|}{$n=60$}  & \multicolumn{2}{c}{$n=80$} \\
      
      \hhline{~|-|-|-|-|-|-|-|-}
      
      & {$\MB$} & {$\MD$} & {$\MB$} & {$\MD$} & {$\MB$} & {$\MD$} & {$\MB$} & {$\MD$}  \\
      \hline
    $p=0.4$ 
    & 0.96 & 0.88 & 1.0 & 0.88 & 1.0 & 0.88 & 1.0 & 0.88 \\
    $p=0.6$ 
    & 0.98 & 0.92 & 1.0 & 0.92 & 1.0 & 0.91 & 1.0 & 0.92 \\
    $p=0.8$ 
    & 0.99 & 0.93 & 1.0 & 0.93 & 1.0 & 0.93 & 1.0 & 0.93 \\
    
    \midrule[1.5pt]

  \multicolumn{9}{c}{$m=70$}\\
    \midrule
    \multirow{2}{*}{} &
      \multicolumn{2}{c|}{$n=20$} &
      \multicolumn{2}{c|}{$n=40$} &
      \multicolumn{2}{c|}{$n=60$} & \multicolumn{2}{c}{$n=80$}  \\
      \hhline{~|-|-|-|-|-|-|-|-}
      & {$\MB$} & {$\MD$} & {$\MB$} & {$\MD$} & {$\MB$} & {$\MD$} & {$\MB$} & {$\MD$} \\
      \hline
    $p=0.4$  & 0.88 & 0.89 & 1.0 & 0.90 & 1.0 & 0.88 & 1.0 & 0.89  \\
    $p=0.6$  & 0.87 & 0.90 & 1.0 & 0.89 & 1.0 & 0.89 & 1.0 & 0.90  \\
    $p=0.8$  & 0.91 & 0.90 & 1.0 & 0.90 & 1.0 & 0.90 & 1.0 & 0.90  \\
    
    \bottomrule[2pt]

  \end{tabular}

\caption{\small Every column labeled with $\MB$ ($\MD$, respectively) reports the average ratio between the payoff the first agent gets from $\MB$ ($\MD$) when it reports truthfully and the payoff that it gets from its \safepolicy. 
We consider $m=30,50,70$, $n=20,40,60,80$, and fix $b_i=3$ for each agent.
The averages are taken over 250 randomly generated instances.
\label{tab_comp:1}
}

\end{table*}

\begin{table*}[ht!]
\centering
    \begin{tabular}{l |c|c|c|c |c|c|c|c}
    \toprule[2pt]
    \multicolumn{9}{c}{$m=30$}\\
    \midrule
    \multirow{2}{*}{} & \multicolumn{4}{c|}{$n=20$} &
      \multicolumn{4}{c}{$n=40$} \\
      \hhline{~|-|-|-|-|-|-|-|-}
      & \multicolumn{2}{c|}{$\MB$} & \multicolumn{2}{c|}{$\MD$} & \multicolumn{2}{c|}{$\MB$} & \multicolumn{2}{c}{$\MD$}   \\
      \hhline{~|-|-|-|-|-|-|-|-}
         & $\max$ & $\min$ & $\max$ & $\min$ & $\max$ & $\min$ & $\max$ & $\min$ \\
      \hline
    $p=0.4$  & 0.12 & 0.0 & 0.28 & 0.0 & 0.0 & 0.0 & 0.30 & 0.04  \\
    $p=0.6$   & 0.0 & 0.0 & 0.24 & 0.0 & 0.0 & 0.0 & 0.26 & 0.06  \\
    $p=0.8$   & 0.0 & 0.0 & 0.19 & 0.06 & 0.0 & 0.0 & 0.21 & 0.06   \\
    
    \midrule[1.5pt]
  
    \multicolumn{9}{c}{$m=50$}\\
    \midrule
    
    \multirow{3}{*}{} & \multicolumn{4}{c|}{$n=20$} &
      \multicolumn{4}{c}{$n=40$} \\
      \hhline{~|-|-|-|-|-|-|-|-}
      &  \multicolumn{2}{c|}{$\MB$} & \multicolumn{2}{c|}{$\MD$} & \multicolumn{2}{c|}{$\MB$} & \multicolumn{2}{c}{$\MD$}   \\
      \hhline{~|-|-|-|-|-|-|-|-}
        & $\max$ & $\min$ & $\max$ & $\min$ & $\max$ & $\min$ & $\max$ & $\min$  \\
      \hline
    $p=0.4$  & 0.26 & 0.0 & 0.24 & 0.02 &  0.0 & 0.0 & 0.24 & 0.03 \\
    $p=0.6$   & 0.25 & 0.0 & 0.16 & 0.02 & 0.0 & 0.0 & 0.21 & 0.02 \\
    $p=0.8$  & 0.14 & 0.0 & 0.13 & 0.02 & 0.0 & 0.0 & 0.10 & 0.03  \\
    
    \midrule[1.5pt]

    \multicolumn{9}{c}{$m=70$}\\
    \midrule
    \multirow{2}{*}{} &
      \multicolumn{4}{c|}{$n=20$} &
      \multicolumn{4}{c}{$n=40$} \\
      \hhline{~|-|-|-|-|-|-|-|-}
       & \multicolumn{2}{c|}{$\MB$} & \multicolumn{2}{c|}{$\MD$} & \multicolumn{2}{c|}{$\MB$} & \multicolumn{2}{c}{$\MD$}  \\
      \hhline{~|-|-|-|-|-|-|-|-}
      & $\max$ & $\min$ & $\max$ & $\min$ & $\max$ & $\min$ & $\max$ & $\min$ \\
      \hline
    $p=0.4$   & 0.39 & 0.0 & 0.23 & 0.02 & 0.0 & 0.0 & 0.21 & 0.04  \\
    $p=0.6$  & 0.41 & 0.0 & 0.18 & 0.02 & 0.0 & 0.0 & 0.17 & 0.04  \\
    $p=0.8$ & 0.41 & 0.0 & 0.15 & 0.04 & 0.0 & 0.0 & 0.15 & 0.04   \\
    
    \bottomrule[2pt]
  \end{tabular}

\caption{Every column labeled with $\MB$ ($\MD$, respectively) reports the maximum and minimum difference between what the first agent gets from $\MB$ ($\MD$) when it reports truthfully and the payoff that it can achieve by applying its \safepolicy.
The average is taken on $250$ randomly generated instances.
\label{tab_comp:2} }

\end{table*}

\begin{table*}[ht!]
\centering
    \begin{tabular}{l |c|c|c|c |c|c|c|c}
    \toprule[2pt]
    \multicolumn{9}{c}{$m=30$}\\
    \midrule
    \multirow{2}{*}{} & \multicolumn{4}{c|}{$n=60$} &
      \multicolumn{4}{c}{$n=80$} \\
      \hhline{~|-|-|-|-|-|-|-|-}
      & \multicolumn{2}{c|}{$\MB$} & \multicolumn{2}{c|}{$\MD$} & \multicolumn{2}{c|}{$\MB$} & \multicolumn{2}{c}{$\MD$}   \\
      \hhline{~|-|-|-|-|-|-|-|-}
         & $\max$ & $\min$ & $\max$ & $\min$ & $\max$ & $\min$ & $\max$ & $\min$ \\
      \hline
    $p=0.4$  & 0.0 & 0.0 & 0.31 & 0.0 & 0.0 & 0.0 & 0.32 & 0.01   \\
    $p=0.6$   & 0.0 & 0.0 & 0.24 & 0.06 & 0.0 & 0.0 & 0.24 & 0.06  \\
    $p=0.8$   & 0.0 & 0.0 & 0.16 & 0.06 & 0.0 & 0.0 & 0.18 & 0.06   \\
    
    \midrule[1.5pt]
  
    \multicolumn{9}{c}{$m=50$}\\
    \midrule
    
    \multirow{3}{*}{} & \multicolumn{4}{c|}{$n=60$} &
      \multicolumn{4}{c}{$n=80$} \\
      \hhline{~|-|-|-|-|-|-|-|-}
      &  \multicolumn{2}{c|}{$\MB$} & \multicolumn{2}{c|}{$\MD$} & \multicolumn{2}{c|}{$\MB$} & \multicolumn{2}{c}{$\MD$}   \\
      \hhline{~|-|-|-|-|-|-|-|-}
        & $\max$ & $\min$ & $\max$ & $\min$ & $\max$ & $\min$ & $\max$ & $\min$  \\
      \hline
    $p=0.4$  & 0.0 & 0.0 & 0.24 & 0.02 &  0.0 & 0.0 & 0.26 & 0.03 \\
    $p=0.6$   & 0.0 & 0.0 & 0.16 & 0.02 & 0.0 & 0.0 & 0.21 & 0.02 \\
    $p=0.8$  & 0.0 & 0.0 & 0.13 & 0.02 & 0.0 & 0.0 & 0.13 & 0.02 \\
    
    \midrule[1.5pt]

    \multicolumn{9}{c}{$m=70$}\\
    \midrule
    \multirow{2}{*}{} &
      \multicolumn{4}{c|}{$n=60$} &
      \multicolumn{4}{c}{$n=80$} \\
      \hhline{~|-|-|-|-|-|-|-|-}
       & \multicolumn{2}{c|}{$\MB$} & \multicolumn{2}{c|}{$\MD$} & \multicolumn{2}{c|}{$\MB$} & \multicolumn{2}{c}{$\MD$}  \\
      \hhline{~|-|-|-|-|-|-|-|-}
      & $\max$ & $\min$ & $\max$ & $\min$ & $\max$ & $\min$ & $\max$ & $\min$ \\
      \hline
    $p=0.4$   & 0.0 & 0.0 & 0.23 & 0.02 & 0.0 & 0.0 & 0.25 & 0.03  \\
    $p=0.6$  & 0.0 & 0.0 & 0.18 & 0.02 & 0.0 & 0.0 & 0.18 & 0.04  \\
    $p=0.8$ & 0.0 & 0.0 & 0.15 & 0.04 & 0.0 & 0.0 & 0.14 & 0.04   \\
    
    \bottomrule[2pt]
  \end{tabular}

\caption{(continued) Every column labeled with $\MB$ ($\MD$, respectively) reports the maximum and minimum difference between what the first agent gets from $\MB$ ($\MD$) when it reports truthfully and the payoff that it can achieve by applying its \safepolicy.
The average is taken on $250$ randomly generated instances.
\label{tab_comp:2_cont} }
\end{table*}

\subsection{The Manipulability of \texorpdfstring{$\MB$}{MB}.}

We now move our focus on the $\MB$.
In particular, we want to assess the manipulability of the $\MB$ from every agent.
For this reason, we consider three metrics:
\begin{itemize}
    \item the Maximum Percentage Utility Gain (MPUG), which measures the maximum gain that any agent can get from manipulating;
    \item the Percentage of Manipulative Agents (PMA), which measures the average number of agents that can manipulate the mechanism, and
    \item the Percentage of Manipulable Instances (PMI) of the mechanism, which measures the number of instances that are manipulable by at least one agent.
\end{itemize}
\begin{figure}[t]
     \centering
     \begin{subfigure}[b]{0.47\textwidth}
         \centering
         \centering
  \includegraphics[width=1.1\linewidth]{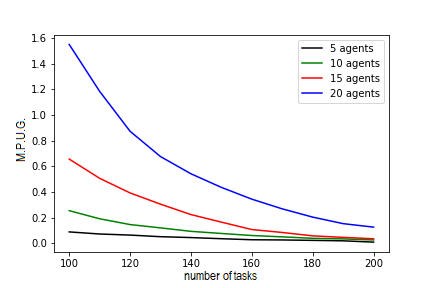}
  \caption{The capacity of every agent is sampled uniformly in $\{3,4,\dots,7\}$. \label{fig:good_2}}
  
     \end{subfigure}
     \hfill
     \begin{subfigure}[b]{0.47\textwidth}
         \centering
      \includegraphics[width=1.1\linewidth]{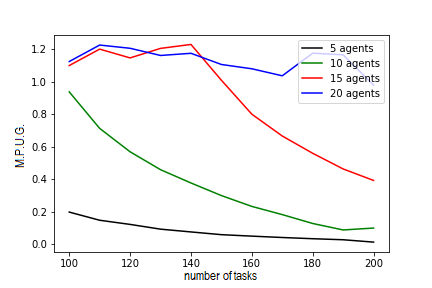}
      \caption{The capacity of every agent is sampled uniformly in 
 $\{7,8,\dots,12\}$. \label{fig:test2}}
      
     \end{subfigure}

        \caption{The Maximum Percentage Utility Gain (M.P.U.G.) for different number of agents, number of tasks and when the agents' capacity ranges in different sets.}
\end{figure}

{\bf Maximum Percentage Utility Gain.} 
The experiment is run as follows.
We generate a random input $G$ for $n=20,40,60,80$, $m=30,50,70$, and $p=0.4,0.6,0.8$.
We then compute the utilities of all the agents in the truthful input.
On every instance and for every agent, we consider the $T$-level manipulations, i.e. every agent $a_i$ hides the edges whose value is less than or equal to $T=1.5,2,2.5,3$.
For every manipulation, we then solve the instance again and compute the utility achieved by $a_i$ after the manipulation. 
We then compute the Percentage Utility Gain of agent $a_i$ as 
\begin{equation}
\label{eq:perc_gain_exp}
p^{(T)}_{a_i}(G):=\frac{\max\{w_i(G')-w_i(G), 0\}}{w_i(G)},
\end{equation}
where $T$ is the threshold defining the manipulation, $G$ is the truthful instance, and $G'$ the instance after agent $a_i$ applied its $T$-level manipulation.
Notice that, by Lemma \ref{lemma:nopayoff}, this quantity is well defined and finite\footnote{we adopt the convention $\frac{0}{0}=1$.}.
Finally, we set the MPUG of the instance $G$ as 
\[
M_p(G):=\max_{T,i}\big\{p^{(T)}_{u_i}(G)\big\}.
\]
The value $M_p(G)$ describes the maximal percentage utility gain that an agent can get by accomplishing one of the above-described $T$-level manipulations. 
We estimate the extent to which an agent can increase its utility percentage by computing the average of the maximal percentage utility gain over $250$ randomly generated instances. 
In Figure \ref{fig:good_2}, we report the maximum percentage utility gain when the capacity of every agent is randomly chosen in the set $\{3,4,5,6,7\}$, while, in Figure \ref{fig:test2}, the capacity is randomly chosen in the set $\{7,8,9,10,11,12\}$.
We observe that the maximum gain that any agent can accomplish constantly drops as the number of tasks grows in both cases.
In Figure \ref{fig:test2}, although there is a trend that the curves are dropping, two of them (15 agents and 20 agents) are initially unstable.
Going back to the worker-project example, these experimental results suggest that when the number of projects is enough to ensure every worker will be associated with its maximum number of tasks, hiding the connections with low-value tasks always leads to a higher utility.
On the contrary, when there is a shortage of projects, hiding the connections with low-value tasks may lead to a loss since it is no longer certain that each worker will be associated with a number of projects equal to its capacity.
It is noteworthy that, when the capacity of the workers is a random integer between $7$ and $12$, the average capacity of the workers is $9.5$. 
If we multiply this average by $15$, we find $142.5$ which is the threshold after which the curve with $15$ agents becomes smooth.
{\bf Percentage of Manipulative Agents.} 
We now run a test to assess the percentage of agents that can benefit from manipulation, i.e. the number of agents that are able to manipulate over the total number of agents.  
The parameters are fixed as follows.
For each triplet of $n=10,15,20$, $m=100,125,150,175,200$, and $p=0.1,0.2,0.4$, we generate $250$ instances and compute the number of manipulative agents for each of those. 
The capacity of every agent is randomly sampled in a discrete set, in our tests we consider the cases in which $\{1,2,3,4,5\}$, $\{3,4,5,6,7\}$, and $\{5,6,7,8,9,10\}$.
Finally, for every agent, we consider the $T$-level manipulations where $T=1.5,2,2.5,3$.
In Table \ref{tab:2}, we report the average PMA for each triplet, i.e. the average percentage of agents that are able to get a higher payoff by applying one of the previously described $T$-level manipulations.  
The results are in accordance with the ones observed for the MPUG: when the number of tasks is higher than the total capacity of the agents, the number of manipulative agents and manipulable instances decreases. 
Similarly, when the number of tasks is lower than the total capacity of the agents, the number of manipulable instances starts to decrease as well.
Once again, this latter effect is due to the shortage of supply for the agents, which again makes the tasks with a low value more appealing.

\begin{table}[ht!]
    \centering
\resizebox{0.8\textwidth}{!}{  
  \begin{tabular}{c|c|c|c|c|c|c|c|c|c}
    \toprule[2pt]
    \multicolumn{10}{c}{$b \in \{1,2,3,4,5\}$} \\
    \midrule
    \multirow{2}{*}{} &
      \multicolumn{3}{c|}{$n=10$} &
      \multicolumn{3}{c|}{$n=15$} &
      \multicolumn{3}{c}{$n=20$} \\
      \hline
      $p$ & {$0.1$} & {$0.2$} & {$0.4$}& {$0.1$} & {$0.2$} & {$0.4$}& {$0.1$} & {$0.2$} & {$0.4$}\\
      \hline
    $m= 100$ & 0.26 & 0.0 & 0.0 & 0.52 & 0.02 & 0.01 & 0.81 & 0.41 & 0.38\\
    $m= 125$ & 0.05 & 0.0 & 0.0 & 0.13 & 0.0 & 0.0 & 0.36 & 0.03 & 0.01\\
    $m=150$ & 0.01 & 0.0 & 0.0 & 0.01 & 0.0 & 0.0 & 0.06 & 0.0 & 0.0\\
    $m=175$ & 0.0 & 0.0 & 0.0 & 0.0 & 0.0 & 0.0 & 0.01 & 0.0 & 0.0 \\
    $m=200$ & 0.0 & 0.0 & 0.0 & 0.0 & 0.0 & 0.0 & 0.0 & 0.0 & 0.0 \\

    \midrule[1.5pt]
    
    \multicolumn{10}{c}{$b \in \{3,4,5,6,7\}$} \\
    \midrule
    \multirow{3}{*}{} &
      \multicolumn{3}{c|}{$n=10$} &
      \multicolumn{3}{c|}{$n=15$} &
      \multicolumn{3}{c}{$n=20$} \\
      \hline
      $p$ & {$0.1$} & {$0.2$} & {$0.4$}& {$0.1$} & {$0.2$} & {$0.4$}& {$0.1$} & {$0.2$} & {$0.4$}\\
      \hline
    $m= 100$ & 0.90 & 0.20 & 0.04 & 1.0 & 0.98 & 0.94 & 1.0 & 1.0 & 1.0\\
    $m= 125$ & 0.54 & 0.01 & 0.0 & 0.98 & 0.49 & 0.35 & 1.0 & 0.99 & 0.98\\
    $m=150$ & 0.19 & 0.0 & 0.0 & 0.59 & 0.02 & 0.02 & 0.97 & 0.78 & 0.73\\
    $m=175$ & 0.04 & 0.0 & 0.0 & 0.18 & 0.0 & 0.0 & 0.66 & 0.18 & 0.09 \\
    $m=200$ & 0.01 & 0.0 & 0.0 & 0.03 & 0.0 & 0.0 & 0.13 & 0.01 & 0.01 \\

  
  
    \midrule[1.5pt]
    \multicolumn{10}{c}{$b \in \{5,6,7,8,9,10\}$} \\
    \midrule
    \multirow{2}{*}{} &
      \multicolumn{3}{c|}{$n=10$} &
      \multicolumn{3}{c|}{$n=15$} &
      \multicolumn{3}{c}{$n=20$} \\
      \hline
      $p$ & {$0.1$} & {$0.2$} & {$0.4$}& {$0.1$} & {$0.2$} & {$0.4$}& {$0.1$} & {$0.2$} & {$0.4$}\\
      \hline
    $m=100$ & 0.96 & 1.0 & 0.96 & 0.99 & 1.0 & 1.0 & 0.98 & 1.0 & 1.0\\
    $m=125$ & 0.99 & 0.74 & 0.33 & 1.0 & 1.0 & 1.0 & 1.0 & 1.0 & 1.0\\
    $m=150$ & 0.96 & 0.10 & 0.03 & 1.0 & 0.99 & 0.98 & 1.0 & 1.0 & 1.0\\
    $m=175$ & 0.70 & 0.01 & 0.0 & 1.0 & 0.76 & 0.56 & 1.0 & 1.0 & 1.0 \\
    $m=200$ & 0.35 & 0.0 & 0.0 & 0.93 & 0.18 & 0.09 & 1.0 & 0.99 & 0.98 \\
    \bottomrule[2pt]

  \end{tabular}}
  
  \caption{Average percentages of manipulable instances for different values of $n$, $m$, and randomly generated capacities. The averages are taken over 250 instances randomly generated. The utility vector is sampled from a uniform distribution over $[1,5]$, and the thresholds $T$ are taken in the set $\{1.5,2,2.5,3\}$. }
    \label{tab:1}
\end{table}

{\bf Percentage of Manipulable Instances.}
Lastly, we run a test to determine the average number of instances in which there exists at least an agent that can manipulate the matching process.  
The test specifics are as follows.
For each triplet of $n=10,15,20$, $m=100,125,150,175,200$, and $p=0.1,0.2,0.4$, we generate $250$ instances and compute the number of manipulative agents for each of those. 
The capacity of every agent is randomly sampled in a discrete set, namely $B$, in our tests we consider $B=\{1,2,3,4,5\}$, $\{3,4,5,6,7\}$, and $\{5,6,7,8,9,10\}$.
Finally, to simulate the manipulations of the agents we let them remove the edges connecting them to the tasks with values lower than $T=1.5,2,2.5,$ and $3$.
For each triplet of $n$, $m$, and $p$, we create $250$ instances and compute the number of manipulative agents for each of those.
In Table \ref{tab:1}, we report the average number of instances in which at least an agent was able to manipulate the outcome of the matching in its favour.
According to what we observed for the PMA and MPUG, the average PMI increases as the total capacity approaches the total amount of tasks. 
Similarly, we observe a quick decrease in the percentage as the number of tasks increases.

\begin{table}[ht!]{}
\centering

\resizebox{0.8\textwidth}{!}{  
  \begin{tabular}{c|c|c|c|c|c|c|c|c|c}
    \toprule[2pt]
    \multicolumn{10}{c}{$b \in \{1,2,3,4,5\}$} \\
    \midrule
    \multirow{2}{*}{} &
      \multicolumn{3}{c|}{$n=10$} &
      \multicolumn{3}{c|}{$n=15$} &
      \multicolumn{3}{c}{$n=20$} \\
      \hline
      $p$ & {$0.1$} & {$0.2$} & {$0.4$}& {$0.1$} & {$0.2$} & {$0.4$}& {$0.1$} & {$0.2$} & {$0.4$}\\
      \hline
    $m=100$ & 0.04 & 0.0 & 0.0 & 0.09 & 0.0 & 0.0 & 0.22 & 0.11 & 0.08\\
    $m=125$ & 0.01 & 0.0 & 0.0 & 0.01 & 0.0 & 0.0 & 0.05 & 0.01 & 0.01\\
    $m=150$ & 0.0 & 0.0 & 0.0 & 0.0 & 0.0 & 0.0 & 0.01 & 0.0 & 0.0\\
    $m=175$ & 0.0 & 0.0 & 0.0 & 0.0 & 0.0 & 0.0 & 0.01 & 0.0 & 0.0 \\
    $m=200$ & 0.0 & 0.0 & 0.0 & 0.0 & 0.0 & 0.0 & 0.0 & 0.0 & 0.0 \\

    \midrule[1.5pt]
    
    \multicolumn{10}{c}{$b \in \{3,4,5,6,7\}$} \\
    \midrule
    \multirow{3}{*}{} &
      \multicolumn{3}{c|}{$n=10$} &
      \multicolumn{3}{c|}{$n=15$} &
      \multicolumn{3}{c}{$n=20$} \\
      \hline
      $p$ & {$0.1$} & {$0.2$} & {$0.4$}& {$0.1$} & {$0.2$} & {$0.4$}& {$0.1$} & {$0.2$} & {$0.4$}\\
      \hline
    $m= 100$ & 0.33 & 0.06 & 0.01 & 0.59 & 0.62 & 0.44 & 0.59 & 0.83 & 0.66\\
    $m= 125$ & 0.14 & 0.01 & 0.0 & 0.44 & 0.15 & 0.08 & 0.71 & 0.66 & 0.50\\
    $m=150$ & 0.04 & 0.0 & 0.0 & 0.17 & 0.01 & 0.01 & 0.54 & 0.28 & 0.20\\
    $m=175$ & 0.01 & 0.0 & 0.0 & 0.04 & 0.0 & 0.0 & 0.21 & 0.04 & 0.01 \\
    $m=200$ & 0.01 & 0.0 & 0.0 & 0.01 & 0.0 & 0.0 & 0.03 & 0.01 & 0.0 \\

  
  
    \midrule[1.5pt]
    \multicolumn{10}{c}{$b \in \{5,6,7,8,9,10\}$} \\
    \midrule
    \multirow{2}{*}{} &
      \multicolumn{3}{c|}{$n=10$} &
      \multicolumn{3}{c|}{$n=15$} &
      \multicolumn{3}{c}{$n=20$} \\
      \hline
      $p$ & {$0.1$} & {$0.2$} & {$0.4$}& {$0.1$} & {$0.2$} & {$0.4$}& {$0.1$} & {$0.2$} & {$0.4$}\\
      \hline
    $m= 100$ & 0.37 & 0.82 & 0.58 & 0.29 & 0.81 & 0.76 & 0.20 & 0.60 & 0.59\\
    $m= 125$ &  0.57 & 0.39 & 0.33 & 0.57 & 0.91 & 0.72 & 0.45 & 0.79 & 0.68\\
    $m=150$ & 0.49 & 0.04 & 0.03 & 0.79 & 0.75 & 0.49 & 0.72 & 0.92 & 0.72 \\
    $m=175$ & 0.27 & 0.01 & 0.0 & 0.73 & 0.35 & 0.17 & 0.89 & 0.85 & 0.62 \\
    $m=200$ & 0.08 & 0.0 & 0.0 & 0.49 & 0.05 & 0.02 & 0.85 & 0.68 & 0.45 \\
    \bottomrule[2pt]
  \end{tabular}}
    \caption{Average percentages of manipulative agents for different values of $n$, $m$, and randomly generated capacities. The averages are taken over 250 instances randomly generated. The utility vector is sampled from a uniform distribution over $[1,5]$, and the thresholds $T$ are taken in the set $\{1.5,2,2.5,3\}$. }
    \label{tab:2}
\end{table}

\subsection{A Randomized Version of \texorpdfstring{$\MB$}{MB}.}

Finally, we consider a version of $\MB$ that processes the agents in a random order, which we call $\MBR$.
The routine of $\MBR$ is the following:
\begin{enumerate}[label=(\roman*)]
    \item the mechanism gathers the reports of the agents (in this case the edges);
    \item the mechanism randomly orders the agents, that is randomly decides which agent is $a_1$, which is $a_2$, and so on, according to a routine $r$;
    \item after the agent set is ordered, it runs Algorithm \ref{alg:Maxvalue} in order to retrieve a $b$-matching.
\end{enumerate}
Since rearranging the order of the agents does not affect the optimality of the output of Algorithm \ref{alg:Maxvalue}, $\MBR$ always returns an optimal $b$-matching.
Since $\MBR$ is not deterministic, we study whether $\MBR$ is truthful in expectation.
If $I_i$ is the truthful report of agent $i\in [n]$, $\MBR$ is truthful in expectation if it holds
\[
\mathbb{E}[w(\MBR(I_i,I_{-i}))]\ge \mathbb{E}[w(\MBR(I_i',I_{-i}))]\quad\quad \forall I_i'\in \mathcal{S}_i, 
\]
where the expected value is taken over all the possible ordering of agents and weighted by their likelihood.
As we observed in Section \ref{sect:strategy}, knowing the agents' order is essential for each agent to determine its best strategy.
Randomizing the processing order before running the routine of Algorithm \ref{alg:Maxvalue} is therefore a natural way to preclude this information from the agents and reduce the overall manipulability of the mechanism.
It is worth stressing, however, that there does not exist a shuffling that makes any optimal mechanism truthful in expectation.
\begin{theorem}
    The mechanism $\MBR$ is not truthful in expectation regardless of the routine that randomizes the order of the agents during step $(ii)$ of the routine.
\end{theorem}

\begin{proof}
    Again, consider the case in which we have two agents, namely $a_1$ and $a_2$, whose capacity is equal to $1$, thus $b_1=b_2=1$.
    Let us consider that we have two tasks, namely $t_1$ and $t_2$, whose values are $q_1=2$ and $q_1=1$.
    Finally, let us assume that in the truthful input, both agents are connected to both the tasks.

    Without loss of generality, we can assume that one of the two agents has a probabilty of getting $t_1$ that is lower than $1$, so that its expected payoff is lower than $2$.
    For the sake of simplicity, assume that such agent is $a_1$.
    Then, regardless of how we shuffle the agents, if $a_1$ hides the edge connecting it to $t_2$, the only possible optimal vertex weighted matching is the one that associate $a_1$ to $t_1$, thus its expected payoff when it manipulates is equal to $2$, which is higher than what it was getting by reporting truthfully.
\end{proof}

While randomizing the order of agents cannot guarantee the truthfulness of the mechanism, it does impact the expected payoff for every agent.
In the following, we consider only the following randomization method.
Given the reports of the agents, every agent $a\in A$ is associated with a value, namely $p_a$ defined as $p_a=\sum_{t_j\in T_a}(1+q_j)^{-1}$, where $q_j$ is the value of task $t_j$.
Once every $p_a$ is computed, we define a probability measure, namely $\eta_a$, over $A$ by normalizing the values $\{p_a\}_{a\in A}$.
We then run a lottery to decide which agent will be the first agent $a_1$.
The probability that agent $a$ is labelled as the first agent is $\eta_a$.
Once the first agent is elected, we remove it from $A$ and run the same procedure over the set of remaining agents.
The procedure goes on until all the agents are labelled.
Among the various shuffling methods, we opted for the one we described above because it punishes agents for hiding edges connecting them to low-valued tasks.
To assess how shuffling the agents' set impacts the performances of $\MB$, we run $\MB$ and $\MBR$ on $100$ randomly generated instances.
For each mechanism, we compute the percentage of instances in which at least one agent is able to manipulate the allocation procedure.
To determine whether an instance is manipulable for $\MB$, we first run the mechanism over the truthful input (i.e. the randomly generated instance) and compute every agent's payoff.
Afterwards, we test if at least one agent is able to increase its payoff by misreporting its private information. 
Thus, for every agent $a_i$, we run $\MB$ over the instance we get when $a_i$ hides the edges connecting it to the tasks with the lowest $2$, $3$, and $4$ values.
For the first agent, we bypass this procedure and simply compare the payoff that the first agent gets with the payoff it would get by using the strategy described in Theorem \ref{thm:max_match_agent_one}.
An instance is considered manipulable if at least one agent increases its payoff by hiding some of its connections.
An instance is manipulable for $\MBR$ if the average payoff that each agent gets by reporting truthfully is higher than the average payoff that the agent gets by hiding its connection with the tasks with the lowest $2$, $3$, and $4$ values.
Every average is taken by running $\MBR$ on the same instance, be it manipulated or not, for $250$ times.
In Table \ref{tab_comp:rand_all_ag_b2}, we report the results of our tests when $n=15,20,25$, $m=25,30$, $p=0.2,0.3$, and the capacity of each agent is equal to $3$.
Finally, to simulate the manipulations of the agents we consider the $k$-th order manipulation for $k=2,3,4$.
Thus, every agent hides the connection to the two, three, and four tasks with the lowest value.
We observe that $\MBR$ is less manipulable than its deterministic counterpart, showing that knowing the processing order of the mechanism is essential information to the manipulative agents.
Interestingly, in most cases, the percentage of manipulable instances is $0$, so there is no agent able to manipulate any of the $100$ instances. 
To further enhance our findings, in Table \ref{tab_comp:rand_only_2} we compare the percentage of instances in which one agent was able to manipulate $\MBR$ with the number of instances in which the first agent is able to manipulate $\MB$.
We report the percentages in Table \ref{tab_comp:rand_only_2}.
We observe that, even if we restrict our attention to just the first agent, $\MB$ is more manipulable than its randomized counterpart. 

\begin{table*}[t]
\centering
\small

    \begin{tabular}{c| c|c| c|c |c|c }
    \toprule[2pt]
    \multicolumn{7}{c}{$m=25$}   \\
    \midrule
      & \multicolumn{2}{c|}{$n=15$} & \multicolumn{2}{c|}{$n=20$} & \multicolumn{2}{c}{$n=25$} \\
      \hhline{~|-|-|-|-|-|-}
       & {\;$\MBR$\;} & {\;$\MB$\;} & {\;$\MBR$\;} & {\;$\MB$\;} & {\;$\MBR$\;} & {\;$\MB$\;}  \\
      \hline
      $p=0.2$\;  & 0.0 & 0.33 & 0.0 & 0.19 & 0.0 & 0.10 \\ 
      $p=0.3$\;  & 0.01 & 0.40 & 0.0 & 0.09 & 0.0 & 0.02 \\

      \midrule[1.5pt]
      
      \multicolumn{7}{c}{$m=30$}\\
      \midrule
      & \multicolumn{2}{c|}{$n=15$} & \multicolumn{2}{c|}{$n=20$} & \multicolumn{2}{c}{$n=25$} \\
      \hhline{~|-|-|-|-|-|-}
      & {\;$\MBR$\;} & {\;$\MB$\;} & {\;$\MBR$\;} & {\;$\MB$\;} & {\;$\MBR$\;} & {\;$\MB$\;}  \\
      \hline
      $p=0.2$\;   & 0.0 & 0.56 & 0.0 & 0.34 & 0.0 & 0.34 \\ 
      $p=0.3$\;   & 0.06 & 0.71 & 0.05 & 0.23 & 0.0 & 0.08 \\
     
    \bottomrule[2pt]
  \end{tabular}
    \caption{Every column labeled with $\MBR$ ($\MB$, respectively) reports the percentage of instances (over $100$) in which at least one agent is able to manipulate $\MBR$ ($\MB$).
    We consider $m=25,30$, $n=15,20,25$, and fix $b_i=b=3$ for every agent.
    The values $q_j$ are sampled from a Gaussian distribution $\mathcal{N}(3,0.77)$.\label{tab_comp:rand_all_ag_b2} } 
\end{table*}

\begin{table*}[t]
\centering
\small

    \begin{tabular}{c| c|c| c|c |c|c }
    \toprule[2pt]
    \multicolumn{7}{c}{$m=25$}   \\
    \midrule
      & \multicolumn{2}{c|}{$n=15$} & \multicolumn{2}{c|}{$n=20$} & \multicolumn{2}{c}{$n=25$} \\
      \hhline{~|-|-|-|-|-|-}
       & {\;$\MBR$\;} & {\;$\MB$\;} & {\;$\MBR$\;} & {\;$\MB$\;} & {\;$\MBR$\;} & {\;$\MB$\;}  \\
      \hline
      $p=0.2$\;  & 0.03 & 0.12 & 0.01 & 0.08 & 0.0 & 0.01 \\
      $p=0.3$\;  & 0.02 & 0.16 & 0.01 & 0.05 & 0.0 & 0.0  \\

      \midrule[1.5pt]
      
      \multicolumn{7}{c}{$m=30$}\\
      \midrule
      & \multicolumn{2}{c|}{$n=15$} & \multicolumn{2}{c|}{$n=20$} & \multicolumn{2}{c}{$n=25$} \\
      \hhline{~|-|-|-|-|-|-}
      & {\;$\MBR$\;} & {\;$\MB$\;} & {\;$\MBR$\;} & {\;$\MB$\;} & {\;$\MBR$\;} & {\;$\MB$\;}  \\
      \hline
      $p=0.2$\;  & 0.03 & 0.25 & 0.01 & 0.12 & 0.0 & 0.08 \\
      $p=0.3$\;  & 0.04 & 0.23 & 0.01 & 0.06 & 0.0 & 0.04 \\
     
    \bottomrule[2pt]
  \end{tabular}
    \caption{\small Additional results of the tests on mechanism $\MBR$. 
    Every column labeled with $\MBR$ reports the percentage over $250$ randomly generated instances in which an agent is able to manipulate $\MBR$. 
    For every instance, we evaluate the manipulability of $\MBR$ by taking the average of $250$ outputs of $\MBR$.
    Every column labeled with $\MB$ reports the percentage, of the same $250$ randomly instances, in which the first agent is able to manipulate $\MB$.
    The capacity of every agent is fixed at $3$, while the tasks' values are sampled from a Gaussian distribution $\mathcal{N}(3,0.77)$.\label{tab_comp:rand_only_2} } 
\end{table*}

\section{Conclusion and Future Work}

In this paper, we studied a game-theoretical framework for the MVbM problems in which one side of the bipartite graph is composed of agents and the other one of tasks that possess an objective value.
We have considered both cases in which the agents and the tasks can behave strategically by either hiding some of the connections with the other side of the bipartite graph or hiding their connections and lowering their capacity/value.
In this framework, we considered three mechanisms, $\MB$, $\MD$, and $\MG$, and fully characterized their truthfulness, approximation ratio, PoA, and PoS. 
For the agents' manipulation case, we have shown that, albeit $\MB$ and $\MD$ are not truthful, they are optimal from a PoA and PoS point of view for all the cases. 
In particular, $\MB$ and $\MD$ are optimal from both a complexity and game theoretical viewpoint.
On the contrary, $\MG$ is truthful and has an approximation ratio equal to $2$, which we proved to be the best possible approximation ratio achievable by truthful mechanisms in this framework.
Remarkably, these results are true in both the EMS and the EMCS.
We have also shown how the ordering of the agents affects the manipulability of the mechanisms and the best/worst Nash Equilibrium they induce and characterize several sets of instances in which $\MD$ and $\MB$ behave truthfully.
We then studied these mechanisms when we allow tasks to behave strategically.
In this simpler framework, all the mechanisms are truthful, thus both $\MB$ and $\MD$ are optimal and truthful.
To conclude, we numerically compared the manipulability of $\MB$ and $\MD$ on randomly generated instances and studied more in detail $\MB$.
A future development of our work would be to further study the manipulability of the mechanisms when we introduce a random shuffling to the set of agents from a theoretical viewpoint.
This study could lead to detect other sets of instances in which the algorithm behaves truthfully.
Another interesting development is to consider multi-layered graphs to better represent different characteristics that the tasks may have.
Finally, it would be interesting to study how limiting the amount of edges that every agent can report affects the performance and manipulability of the mechanism itself.
%









\bibliography{sn-bibliography}

\end{document}